\documentclass[a4paper,UKenglish]{lipics-v2019}

\usepackage{amsthm}
\usepackage{graphicx}
\usepackage{mathtools}
\usepackage{enumerate} 
\usepackage{url}
\usepackage{graphicx}
\usepackage{verbatim}
\usepackage{tikz}
\usepackage{textcomp}
\usepackage{xspace}
\usepackage{amsmath}
\usepackage[all]{xy}

\usepackage{array,multirow}
\usepackage{lineno}

\newcommand{\nat}{\mathbb{N}}
\renewcommand{\implies}{\rightarrow}

\renewcommand{\phi}{\varphi}

\newcommand{\Ps}{\mathcal{P}}
\newcommand{\Var}{\textsf{Var}}

\newcommand{\FV}{\textsf{FV}}

\newcommand{\FO}{\textup{FO}\xspace}
\newcommand{\ML}{\textup{ML}\xspace}
\newcommand{\FOML}{\textup{FOML}\xspace}

\newcommand{\existsBox}[1]{\exists #1 \Box}
\newcommand{\forallBox}[1]{\forall #1 \Box}
\newcommand{\boxExists}[1]{\Box\exists #1}
\newcommand{\boxForall}[1]{\Box\forall #1}
\newcommand{\forallDiamond}[1]{\forall #1 \Diamond}
\newcommand{\existsDiamond}[1]{\exists #1 \Diamond}
\newcommand{\diamondForall}[1]{\Diamond\forall #1}
\newcommand{\diamondExists}[1]{\Diamond\exists #1}

\newcommand{\AB}{\texttt{AB}\xspace}
\newcommand{\EB}{\texttt{EB}\xspace}
\newcommand{\BA}{\texttt{BA}\xspace}
\newcommand{\BE}{\texttt{BE}\xspace}

\newcommand{\ABEB}{\texttt{ABEB}\xspace}
\newcommand{\ABBA}{\texttt{ABBA}\xspace}

\newcommand{\EBBA}{\texttt{EBBA}\xspace}
\newcommand{\EBBE}{\texttt{EBBE}\xspace}
\newcommand{\BABE}{\texttt{BABE}\xspace}

\newcommand{\ABEBBA}{\texttt{ABEBBA}\xspace}
\newcommand{\ABEBBE}{\texttt{ABEBBE}\xspace}
\newcommand{\ABBABE}{\texttt{ABBABE}\xspace}
\newcommand{\EBBABE}{\texttt{EBBABE}\xspace}

\newcommand{\ABEBBABE}{\texttt{ABEBBABE}\xspace}

\newcommand{\LBF}{\texttt{LBF}\xspace}

\newcommand{\M}{\mathcal{M}}
\newcommand{\W}{\mathcal{W}}
\newcommand{\D}{\mathcal{D}}
\newcommand{\R}{\mathcal{R}}
\newcommand{\live}{\delta}
\newcommand{\val}{\rho}

\newcommand{\T}{\mathcal{T}}
\newcommand{\onlyT}{\textsf{OnlyT}}
\newcommand{\Hsuc}{\textsf{Hsuc}}
\newcommand{\Vsuc}{\textsf{Vsuc}}
\newcommand{\children}{\Delta}

\newcommand{\literal}{\textsf{literal}\xspace}
\newcommand{\module}{\textsf{module}\xspace}
\newcommand{\comp}{\textsf{C}}
\newcommand{\Esafe}{\textsf{Existential}-\textsf{safe}\xspace}
\newcommand{\END}{\texttt{END}}
\newcommand{\type}{\Pi}

\newcommand{\Center}{\mathcal{T}_0}
\newcommand{\Left}{\mathcal{T}_1}
\newcommand{\Right}{\mathcal{T}_2}

\newcommand{\EXPspace}{\textsc{ExpSpace}\xspace}
\newcommand{\NEXP}{\textsc{Nexp}\xspace}
\newcommand{\NEXPtime}{\textsc{NexpTime}\xspace}
\newcommand{\PSPACE}{\textsc{PSpace}\xspace}

\usepackage{pifont}
\newcommand{\cmark}{\ding{51}}
\newcommand{\xmark}{\ding{55}}

\title{Are Bundles Good Deals for \FOML?}

\author{Mo Liu}{LORIA, Nancy, France}{mo.liu@loria.fr}{}{}

\author{Anantha Padmanabha}{DI ENS, ENS, CNRS, PSL University \&
  Inria, Paris, France}{anantha.padmanabha@inria.fr}{}{}

\author{R Ramanujam}{Institute of Mathematical Sciences, HBNI , Chennai, India (Retd)\\ Azim Premji University, Bengaluru, India (Visiting)}{jam@imsc.res.in}{}{}

\author{Yanjing Wang}{Department of Philosophy, Peking University, Beijing, China}{y.wang@pku.edu.cn}{}{}

\authorrunning{Padmanabha, Ramanujam}

\Copyright{Anantha Padmanabha and R Ramanujam}



\category{}

\supplement{}

\funding{}

\nolinenumbers 
\hideLIPIcs  

\begin{document}

\maketitle

\begin{abstract}
Bundled products are often offered as good deals to customers. When we bundle quantifiers and modalities together (as in $\exists x \Box$, $\Diamond \forall x$ etc.) in first-order modal logic (\FOML), we get new logical operators whose combinations produce interesting fragments of \FOML without any restriction on the arity of predicates, the number of variables, or the modal scope. It is well-known that finding decidable fragments of \FOML is hard, so we may ask: do bundled fragments that exploit the distinct expressivity of \FOML constitute good deals in balancing the expressivity and complexity? There are a few positive earlier results on some  particular fragments. In this paper, we try to fully map the terrain of bundled fragments of \FOML in (un)decidability, and in the cases without a definite answer yet, we show that they lack the finite model property. Moreover, whether the logics are interpreted over constant domains (across states/worlds) or increasing domains presents another layer of complexity. We also present the \textit{loosely bundled fragment}, which  generalizes the bundles and yet retain decidability (over increasing domain models).
 \end{abstract}

\section{Introduction}
Among the gifts of logic to computer science, perhaps the most important
are the use of first-order logic ($\FO$) for descriptions of computational domains,
and propositional modal logic ($\ML$) for descriptions of state transitions. The former
has led to striking results, not only in database theory, finite model theory,
and descriptive complexity, but also for knowledge representation in artificial 
intelligence. The latter has led to industrial applications, especially in the use
of propositional temporal logics and epistemic logics for system specification and verification.

A natural invitation then is to combine the best of both, using first-order
modal logic ($\FOML$) for the description of state transition systems where states
are given by first-order descriptions of computational domains. Such an
idea is implicit in the realm of database updates, in the control of  
infinite state systems, networks with unbounded parallelism (such as in 
the context of mobile processes), and in the study of cryptographic protocols. In such contexts, there are unboundedly many data elements or processes or, in the case of security protocols, multi-sessions, and we wish to study how these change on applying some transition.
Unfortunately, first-order modal logic seems to combine the worst of both computationally, losing also some good properties of first-order logics and propositional modal logics. 

One of the significant achievements of the last century was mapping the
terrain of decidable fragments of first-order logic (\cite{BGG97}).  Propositional
modal logics are shown  to be robustly decidable (\cite{Vardi1997,AndrekaNB98}), with many extensions
going beyond first-order expressiveness (such as transitive closure) yielding
decidability (cf. e.g., \cite{GMV99}). On the other hand, it is hard to obtain decidable
fragments of $\FOML$, and the situation seems quite hopeless: even the {\em 
two-variable fragment with one single unary predicate} is (robustly) 
undecidable over almost all useful model classes \cite{RybakovS17a}.

Despite such discouragement, a small band of researchers have managed
to identify some decidable fragments of $\FOML$. Notable among them are
the  \textit{monodic fragments} \cite{HodkinsonWZ00,Mono01}, which requires that there be at most  one free 
variable in the scope of any modal subformula. Combining the monodic 
restriction with a decidable fragment of $\FO$  we  often obtain decidable fragments of $\FOML$ even with extra frame conditions. The principal idea here is that two variables occurring free in modal scope allow us to code up a binary relation, which leads to undecidability in the presence of quantifiers, and hence we restrict the use of such variables. Interestingly, this idea originally came from the study on decidable description logics in knowledge representation (cf. \cite{HodkinsonWZ01}). 

This line of research has led to applications in temporal logics for
infinite state systems (\cite{DFKLFOtemporal08}), branching time
temporal logics (\cite{HWZ02}), epistemic logics (\cite{BelardinelliL09,BelardinelliL12}) and logics with counting quantifiers (\cite{Hampson2016-HAMDFM}), and so on. 
  
  \bigskip
    
In studying decidable syntactic fragments of $\FO$, we typically study 
restrictions on  quantifier alternation, vocabulary or the number of variables in the formula. The use of equality and constants often brings some complexity as well. When it comes to guarded fragments \cite{AndrekaNB98}, fluted fragments \cite{tendera2019} and such, there are no restrictions on variables, quantification
or vocabulary, but on \textit{quantification scope}: quantification of a variable is not free, but is subject to its position in the syntax tree of the formula. Essentially, this can limit the power of the quantifiers, which is also the `secret of success' of propositional modal logic computationally.
These are
reminiscent of `quantifier phrases' in linguistics, where quantifiers are 
combined with predicates or connectives (as in `both' or `all up'). 

A recent idea of this kind was explored by  Wang (\cite{Wang17d})
who showed that when the existential quantifier and a box modality 
were always bundled together to appear as a single quantifier-modality pair ($\exists x\Box$), the resulting fragment of FOML enjoys the attractive properties of propositional modal logic: finite tree model property, \PSPACE decision procedure, simple axiom system and so on, without any restriction on predicates or the occurrences of variables. It does not build on a decidable \FO fragment explicitly, but exploits the distinct feature of \FOML in capturing the interaction between quantifiers and modalities, as in Barcan formulas. Interestingly, similar to the origin of monodic fragments, the motivation again came from a formal treatment of knowledge. Wang used $\exists x\Box$ as a new modality where $\Box$ is the know-that operator to capture the logic structure behind the \textit{knowing-wh} expressions such as \textit{knowing what, knowing how, knowing why}, and so on (cf. \cite{Wang2018}), e.g.,  \textit{knowing how to achieve $\phi$} is rendered as \textit{there exists} a method $x$ such that the agent \textit{knows that} $x$ can guarantee $\phi$ \cite{Wang18}. 

Encouraged by this, \cite{padmanabha_et_al:LIPIcs:2018:9942} took the next step, by considering
not only the combination  $\exists x \Box$ but also its companion 
$\forall x\Box$. They found that the logic with both of these combinations
continued to be decidable over increasing domain models, though it was later shown that there was 
a price to be paid in terms of complexity. Such modal-quantifier combinations
were thus called \textit{bundled modalities} or simply \textit{bundles}. Clearly we can define more bundles such as $\Box \forall$, $\Box \exists$ etc., with the obvious semantics.

Beyond the epistemic context,  bundled fragments also offers many  interesting possibilities for system specification:
\begin{itemize}
\item $\neg \existsBox{x}~(x < c)$: No element is guaranteed to be bounded
by constant $c$ (after update).
\item $\existsBox{x}~\Box \exists{y}~(x > y)$: There is an element that dominates some element after every update.
\item $\Box \exists{x}~\Box \forall{y}~(x \leq y)$: All updates admit a local minimum.
\item $\existsBox{x}~(\existsBox{y}~(x > y) \land \existsBox{y}~(x < y))$: There
is an element that dominates another no matter the update and is dominated
by another no matter the update.
\end{itemize}
The story in \cite{padmanabha_et_al:LIPIcs:2018:9942} came with a twist:  the computational property of many bundled fragments depends on whether you assume the domain of each world is constant or only expanding along the relation.\footnote{Making the domain  vary arbitrarily across the worlds causes problems defining proper semantics for quantifiers cf. \cite{Cresswell96}} 
In the former, a single domain is fixed for the entire model, with only predicate interpretations changing during system evolution. This is natural for many database applications. In the latter, the domain itself may be state-dependent.
This is natural in the context of systems of processes where new
sub-processes are created during system evolution,  and in security theory, where new parallel sessions of security protocols are spawned.  \cite{padmanabha_et_al:LIPIcs:2018:9942}  showed  that while the $\exists x \Box$ 
bundle cannot distinguish between the two interpretations,  even the
fragment of $\forall\Box$ with only unary predicates is undecidable over constant domain models. For other decidable fragments such as the monodic fragment, there is a natural translation showing that for any formula over a varying domain model, there is an equisatisfiable formula over a constant domain model. Unfortunately, such a translation does not preserve expressibility in bundled fragments.

All this opens up a range of questions: what about other bundles such as 
$\Box \exists x$ or $\Box \forall x$ and combinations thereof? Which of these
distinguish constant domain and varying domain models? What about further
bundles such as $\forall x \exists y \Box$? Can we identify the border line
between decidability and undecidability in this terrain? 

This is the project taken up in this paper, and what we present is a 
{\em trichotomy} on the decidability of bundled fragments: decidable ones, undecidable ones, and for those without definite answer yet, we show they lack the finite model property. Moreover, we
present the {\textit loosely bundled} fragment that generalizes the bundling idea to what we believe to be the strongest yet decidable bundled fragment.

\begin{figure}
\begin{center}
\begin{tabular}{|l|l|l|l|l||l|}
\hline
 Domain&$\forall\Box$&$\exists\Box$&$\Box\forall$&$\Box\exists$&Upper/ Lower Bound\\
\hline
\hline
\multirow{5}{*}{Constant}&\cmark &*&*&*&\multirow{2}{*}{Undecidable}\\
\cline{2-5}
& *&*&\cmark &* &\\
\cline{2-6} 
 & \xmark & \cmark & \xmark & \xmark & \PSPACE-complete \\
\cline{2-6}
& \xmark & \xmark & \xmark & \cmark  &\multirow{2}{*}{No FMP}\\
\cline{2-5} 
& \xmark & \cmark  &  \xmark & \cmark   &\\
\hline
\hline
\multirow{11}{*}{Increasing}& \cmark & \xmark & \xmark & \xmark &  \\
\cline{2-5} 

&\xmark & \cmark & \xmark & \xmark &\PSPACE-complete    \\
\cline{2-5} 

&\xmark & \xmark & \cmark & \xmark &   \\
\cline{2-6} 

&\xmark & \xmark & \xmark & \cmark & \EXPspace/ \PSPACE\\

\cline{2-6}
\noalign{\vskip-2\tabcolsep \vskip-3\arrayrulewidth \vskip\doublerulesep}
\\
\cline{2-6} 

&\cmark & \cmark & \xmark & \xmark & \multirow{2}{*}{\EXPspace/\NEXPtime} \\
\cline{2-5} 
&\xmark & \xmark & \cmark & \cmark &\\

\cline{2-6} 

&* & \cmark & \cmark & * & Undecidable \\
\cline{2-6}

&\xmark & \cmark & \xmark & \cmark & No FMP\\
\cline{2-6} 
\noalign{\vskip-2\tabcolsep \vskip-3\arrayrulewidth \vskip\doublerulesep}
\\
\cline{2-6} 

&\cmark & \cmark & \xmark & \cmark & Undecidable\\
\cline{2-6} 

&\cmark & \xmark & \cmark & \cmark & \multirow{2}{*}{\EXPspace/ \NEXPtime} \\
\cline{2-5}
&\multicolumn{4}{c||}{ loosely bundled} &\\
\hline

\end{tabular}
\end{center}
\caption{Satisfiability problem classification for combinations of bundled fragments. The results are new in this paper, except those summarized in Theorem 1 below.
 }
\label{fig-big table}
\end{figure}















The results are presented in Fig \ref{fig-big table}, ordered by the number of various bundles. A $*$ in the figure means that the result holds with or without the presence of the corresponding bundle. Note that not all of the cases are listed explicitly in the table: some cases are covered by others, e.g., the decidability of the fragment with $\forall\Box$ and $\Box\forall$ follows from the case with an extra $\Box\exists$. It is to be noted that constant 
domain  and increasing domain interpretations make a significant difference.  Where the logics are decidable, we present a tableau-based decision procedure. Proofs of undecidability involve coding of tiling problems, and lack of finite model property is shown by forcing infinite domains. The details are tricky, but 
technically interesting, combining many techniques that arise from the study of first-order logics and from modal logics. Specifically, as we will see in Section \ref{sec: ABABE}, for some combinations of bundles we can pull $\exists$ outside the scope of $\forall$ thereby allowing us to search for witnesses that work across worlds, using realised types. 

%

The paper is structured as follows. After presenting the various fragments in the next section, we proceed from bad news to good news, presenting
undecidability results, lack of finite model property and then decidability.

\section{Syntax and Semantics}
The syntax of First Order Modal logic is given by extending the first order logic with modal operators. Note that we exclude equality, constants and function symbols from the syntax.

\begin{definition}[$\FOML$ syntax]
\label{def: FOML syntax}
Given a countable set of predicates $\Ps$  and   a countable set of variables $\Var$, the syntax of $\FOML$ is given by:
$$\alpha ::=  P(x_1,\ldots,x_n)  \mid \neg \alpha \mid \alpha \land \alpha  \mid \exists x \alpha \mid \Box \alpha $$ 
where $P\in \Ps$ has arity $n$ and $x,x_1,\ldots, x_n \in \Var$.
\end{definition}
The boolean connectives $\lor, \implies$, $\leftrightarrow$, and the modal operator $\Diamond$ which is the dual of $\Box$ and the quantifier $\forall$ are all defined in the standard way. The notion of free variables, denoted by $\FV(\alpha)$ is similar to what we have for first order logic with $\FV(\Box\alpha) = \FV(\alpha)$. We write $\alpha(x)$ to mean that $x$ occurs as a free variable of $\alpha$. Also, $\alpha[y/x]$ denotes the formula obtained from $\alpha$ by replacing every free occurrence of $x$ by $y$.

\begin{definition}[$\FOML$ structure]
\label{def: FOML structure}
An increasing domain model for $\FOML$ is a tuple $\M = (\W, \D, \live, \R, \val)$ where, 
$\W$ is a non-empty countable set called {\em  worlds};\footnote{Note that $\FOML$ can be translated into two-sorted $\FO$, and due to the L\"owenheim–Skolem theorem for countable languages, every model has an equivalent countable model, cf. \cite{G07FML}. \label{ft.LS}} $\D$ is a non-empty countable set called {\em domain}; $\R\subseteq (\W\times \W)$ is the {\em accessibility relation}. The map $\live:\W\mapsto 2^\D$ assigns to each 
$w\in \W$ a \textit{non-empty} local domain set such that whenever
$(w,v) \in \R$ we have $\live(w)\subseteq \live(v)$ and 
$\val: (\W\times \Ps) \mapsto \bigcup\limits_{n}2^{\D^n}$ is the {\em valuation function} which specifies the interpretation of predicates at every world over the local domain with appropriate arity.
\end{definition}

The monotonicity condition is useful for evaluating the free variables present in the formula \cite{Cresswell96}. These models are called {\em increasing domain models}. The model $\M$ is said to be a {\em constant domain model} if for all $w\in \W$ we have $\live(w) = \D$. 

For a given model $\M$ we denote $\W^\M,\R^\M$ etc to refer to the corresponding components. We simply use $\W,\R,\live$ etc when $\M$ is clear from the context.

To evaluate formulas, we need an assignment function for variables. For a given model $\M$, an assignment  function $\sigma: \Var\mapsto \D$ is {\em relevant} at 
$w \in \W$ if $\sigma(x)\in \delta(w)$ for all $x\in \Var$.  

\begin{definition}[$\FOML$ semantics]
\label{def: FOML semantics}
Given an $\FOML$ model $\M = (\W, \D, \live, \R, \val)$ and $w \in \W$, and $\sigma$ relevant at $w$, for all $\FOML$ formula $\alpha$
define $\M,w,\sigma \models \alpha$ inductively as follows:
{\small
$$\begin{array}{|lcl|}
\hline
\M, w, \sigma\models P(x_1,\ldots,x_n) &\Leftrightarrow & (\sigma(x_1), \ldots, \sigma(x_n))\in \val(w,P)  \\  
\M, w, \sigma\models \neg\alpha &\Leftrightarrow&   \M, w, \sigma\not\models \alpha \\ 
\M, w, \sigma\models (\alpha\land \beta) &\Leftrightarrow&  \M, w, \sigma\models \alpha \text{ and } \M, w, \sigma\models \beta \\ 
\M, w, \sigma\models \exists x \alpha &\Leftrightarrow& \text{there is some $d\in \live(w)$ such that }  \M, w, \sigma_{[x\mapsto d]}\models \alpha \\
\M, w, \sigma\models \Box  \alpha &\Leftrightarrow&  \text{for every } u\in \W  \text{if $(w,u)\in\R$ then } \M, u, \sigma\models\alpha\\

\hline
\end{array}$$
}
\end{definition}

We sometimes write $\M,w\models \alpha(a)$ to mean $\M,w,[x\mapsto a] \models \alpha(x)$. 

A formula $\alpha$ is \emph{satisfiable} if there is some $\FOML$ structure $\M$ and $w\in \W$ and some assignment $\sigma$ relevant at $w$  such that $\M,w,\sigma \models \alpha$. In the sequel, we will only talk about the relevant $\sigma$ given a pointed model. A formula $\alpha$ is \emph{valid} if $\neg \alpha$ is not satisfiable.

\subsection{Bundled fragments}

The motivation of bundling is to restrict the occurrences of quantifiers using modalities. For instance allowing only formulas of the form $\forall x \Box \alpha$ is one such bundling. We could also have $\Diamond \exists y~\alpha$. Thus, there are many ways to `bundle' the quantifiers and modalities. We call these the `bundled operators'. The following syntax defines all possible bundled operators of one quantifier and one modality:

\begin{definition}[Bundled-$\FOML$ syntax]
\label{def: bundled-FOML syntax}
The bundled fragment of $\FOML$ is the set of all formulas constructed by the following syntax:
{\small
$$\alpha ::= P(x_1,\ldots,x_n)  \mid \neg \alpha \mid \alpha \land \alpha \mid \Box \alpha \mid \existsBox{x} \alpha\mid \forallBox{x} \alpha \mid \boxExists{x} \alpha \mid \boxForall{x} \alpha $$
}
where $P\in \Ps$ has arity $n$ and $x,x_1,\ldots, x_n \in \Var$.
\end{definition}

Note that the dual of the bundled operators will give us the formulas of the form $\forallDiamond{x} \alpha,~\existsDiamond{x}\alpha,~\diamondForall{x} \alpha,~\diamondExists{x} \alpha$. Also, note that $\Box \alpha$ can be defined using any one of the bundled operator where the quantifier is applied to a variable that does not occur in $\alpha$. However, we retain $\Box\alpha$ in the syntax for technical convenience. 

The following \textit{constant domain} models  may help to get familiar with bundles. 
$$
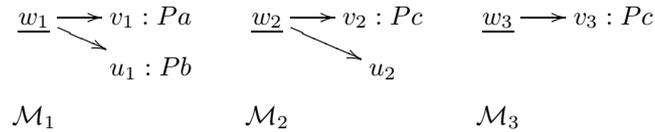
\xymatrix@R-20pt@C-10pt{
\underline{w_1} \ar[r]\ar[dr]& v_1: Pa  &  \underline{w_2} \ar[r]\ar[dr]& v_2: Pc &\underline{w_3} \ar[r]& v_3: Pc  \\
   &                    u_1: Pb    &    & u_2 \\
   \M_1 &&\M_2 & &\M_3
}
$$ 
where $\D^{\M_1}=\{a,b\}$, $\D^{\M_2}=\D^{\M_3}=\{c\}$. $\boxExists{x}Px$ holds at $w_1$ and $w_3$ but not at $w_2$; $\existsBox{x}Px$ holds only at $w_3$; $\neg \forallBox{x}Px$ holds at $w_1$ and $w_2$; $\neg\boxForall{x}\neg Px$ holds at all the $w_i$.   

\medskip
 We denote $\AB$ (to mean forAll-Box) to be the language that allows only atomic predicates, negation, conjunction, $\Box\alpha$ and $\forallBox{x} \alpha$ (dually $\existsDiamond{x} \alpha$) formulas. Similarly we have $\EB$(Exists-Box), $\BA$(Box-forAll) and $\BE$(Box-Exists) to mean the fragments that allows formulas of the form $\exists x \Box\alpha$, $\Box\forall x \alpha$ and $\Box\exists x\alpha$ and their duals respectively. Note that the atomic formulas, negation, conjunction and $\Box\alpha$ are allowed in all the fragments. In general, these fragments are not equally expressive, e.g., as shown by \cite{Wang17d}, the $\EB$ fragment cannot express $\boxExists{x}$, $\forallBox{x}$ and $\boxForall{x}$ bundles over models with increasing (or constant) domains. Complete axiomitizations of the $\EB$-fragment over various frame classes are provided in \cite{Wang17d}\cite{Wang21}.

 We are interested in the combinations of the bundled operators and their decidability.  We denote the combinations of the bundled operators using their short-hand notations. For instance $\EBBA$ (to mean exists-Box + Box-all) is the language that allows atomic predicates, negation, conjunction along with $\existsBox{x}\alpha$ and $\boxForall{x} \alpha$. So in $\EBBA$ we can write formulas of the form $\boxForall{x}\existsBox{y} P(x,y)$, which allows quantifier alternation over the same local domain. 
 
 Similarly we have other combinations like $\ABBA, \EBBABE$ etc. We denote $\ABEBBABE$ to be the fragment that contains all the bundled formulas.\footnote{The naming follows the convention that every pair of letters correspond to a bundled operator and we follow the precedence $\AB > \EB > \BA > \BE$. That is, if the first two letters are not $\AB$ then the fragment does not include  $\forall x\Box$ as a bundle and so on.} 
 
 We summarize the existing results below. 
 \begin{theorem}\label{thm.known}
{\cite{padmanabha_et_al:LIPIcs:2018:9942,Liu2019bundled}}The satisfiability problem for the fragments:
\begin{itemize}
\item  $\ABEB$ and $\BABE$ are decidable over increasing domain models.
\item $\AB$ and $\BA$ are undecidable over constant domain (even with unary predicates).
\item $\EB$  is decidable over constant domain. Moreover $\EB$ cannot distinguish between increasing domain models and constant domain models.
\end{itemize}
 \end{theorem}
 
 \paragraph*{Loosely Bundled Fragment} To study the most general decidable fragment, we introduce the notion of \textit{`loosely bundled fragment'} (\LBF). Note that a bundled formula of the form $\existsBox{x} \alpha$ imposes a restriction that there is exactly one modal formula in the scope of $\exists x$. But this is a strong requirement. We weaken this condition to allow formulas of the form $\exists x\beta$ where $\beta$ is a boolean combination of atomic formulas and modal formulas. Moreover we can allow a quantifier alternation of the form $\exists x_1,\ldots \exists x_n ~\forall y_1,\ldots \forall y_m ~\beta$. As we will see, the fact that the existential quantifiers are outside the scope of universal quantifiers can help us to obtain decidability results over increasing domain models.
 
 \begin{definition}[$\LBF$ syntax]
\label{def-LBF syntax}
The loosely bundled fragment of $\FOML$ is the set of all formulas constructed by the following syntax:

{\small 
\begin{tabular}{r l}
$\psi::=$&$ P(z_1,\ldots z_n)\mid \neg P(z_1,\ldots z_n) \mid \psi\land \psi \mid \psi \lor \psi \mid \Box\alpha \mid \Diamond \alpha$\\
$\alpha:=$&$\psi \mid  \alpha \land \alpha \mid \alpha \lor \alpha \mid \exists x_1\ldots \exists x_k \forall y_1\ldots \forall y_l~\psi$
\end{tabular} 
}

\noindent where $k,l,n \ge 0$ and $P\in \Ps$ has arity $n$ and $x_1,\ldots x_k$,  $y_1,\ldots y_l$, $z_1,\ldots, z_n \in \Var$.
\end{definition}

We let $\LBF$ be the set of all formulas that can be obtained from the grammar of $\alpha$ above. Note that the syntax does not allow a quantifier alternation of the form $\forall x \exists y~ \alpha$. Also, inside the scope of quantifier prefix $\exists^*\forall^*$, we can only have boolean combination of atomic and modal formulas.
 
 Another advantage of the loosely bundled fragment is that some combination of bundled fragments can be embedded into $\LBF$. Hence proving the decidability for $\LBF$ implies the decidability for these combinations as well.
 \begin{proposition}
 \label{prop-LBF embeds many bundled operators}
 The fragments $\ABEB$ and $\BABE$ are  subfragments of  $\LBF$.
 \end{proposition}

\subsection{Trichotomy}
 
The goal of this paper is to classify the decidability status of the satisfiability problem for the various combinations of the bundled fragments.  As discussed before, decidability depends on whether we are considering increasing domain models or constant domain models. 

We consider all possible combinations and obtain a trichotomy classification. We  prove for any combination of bundles, that it is undecidable, or that it admits a tableau based decision procedure or that it lacks the finite model property. The last cases do not give us any (un)decidability result, but demonstrate that the tableau method based on finite model property does not work.

The key idea is that in our setting a fragment can be proved to be undecidable if it can assert $\forall x \exists y \Box~ \alpha$ and $\forall x \Box \forall y \Box \forall z~ \alpha$.  As we will see, the first property is used in the tiling encoding to assert that every `grid point' $x$ has a horizontal/vertical successor $y$. It is important for both quantifiers to be applicable over the same local domain. Moreover  $\Box\alpha$ in $\forall x \exists y \Box \alpha$ ensures that the witness $y$ acts uniformly across all the descendants. The second formula $\forall x \Box \forall y \Box \forall z~ \alpha$ is used to verify the `diagonal property' of the grid. These formulas can be asserted in the fragments $\EBBA$ and $\ABEBBE$ and consequently are shown to be undecidable over increasing domain models.

If a fragment can express $\forall x \exists y \Box~ \alpha$ but not $\forall x \Box \forall y \Box \forall z~ \alpha$ then we will prove that such fragments do not have finite model property. For instance the fragment $\EBBE$ has this property. Also note that over constant domain models, the fragment $\BE$ can assert $\forall x \exists y \Box~\alpha$ in the form $\Diamond \forall x \Box \exists y \Box~ \alpha$. In this case, even though the two quantifiers are applied at different worlds, since the local domain is same across all worlds, the formula will serve the same purpose as intended. We will prove that these fragments do not have the finite model property.

Finally, if a fragment cannot express $\forall x\exists y \Box~ \alpha$ then we will prove that it satisfies finite model property and give a tableau procedure. These fragments include $\ABEBBA$ and $\LBF$.

\section{Undecidable fragments}
\label{sec-undec}

 In \cite{padmanabha_et_al:LIPIcs:2018:9942}  the authors prove that the $\AB$ fragment over constant domain models is undecidable. In \cite{Liu2019bundled} it is proved that a similar undecidability reduction holds for the $\BA$ fragment over constant domain models. 
 In this section, we consider the undecidable bundled fragments over increasing domain models. These are the fragments in which we can express both $\forall x \exists y \Box \alpha$ and $\forall x \Box\forall y \Box \forall z \alpha$ which include $\EBBA$ and $\ABEBBE$ fragments.

For the proof we use the tiling problem over the first quadrant. Given a tiling instance $\T = (T,H,V,t_0)$ where $T$ is a finite set of tiles and $t_0\in T$ and $H,V \subseteq T\times T$ are the horizontal and vertical constraints respectively, a mapping $f: (\nat\times \nat) \to T$ is called a proper tiling if $f(0,0) = t_0$ and for all $i,j \in \nat$ if $f(i,j) = t$ and $f(i,j+1) = t'$ then $(t,t') \in V$ and similarly if $f(i,j) = t$ and $f(i+1,j) = t'$ then $(t,t') \in H$.  The problem, to decide whether an input tiling instance $\T$ has a proper tiling is undecidable \cite{van1996convenience}.

In the encoding, the idea is to interpret domain elements as grid points. We use two binary predicates $P$ and $Q$ to identify the horizontal (going right) and vertical successors (going up) respectively. Also, we abuse notation and consider every $t \in T$ as a unary predicate so that $t(x)$ means the `grid point' $x$ is tiled with $t$. 

\subsection{$\EBBA$ over increasing domain}

This fragment allows formulas of the form $\exists x\Box \alpha$ and $\Box \forall x~\alpha$ (and their duals). Hence we can express $\forall x \exists y \Box \alpha$ and $\forall x \Box\forall y \Box \forall z \alpha$ in the form of $\Box\forall x \exists y \Box \alpha$ and also $\Box\forall x\Box \forall y \Box \forall z \phi$ respectively.
For a given a tiling instance $\T$ we construct an $\EBBA$ sentence $\phi_\T$ such that there exists a proper tiling of $\T$ iff $\phi_T$ is satisfiable in an increasing domain model.
We first define some notations for short-hand reference:

\begin{figure}[h]
  \centering
  \small
  \begin{tabular}{|r c l |}
\hline &&\\
$\onlyT(t,x)$&$:=$&$t(x) \land \bigwedge\limits_{t'\ne t} \neg t'(x)$\\
\hline &&\\
$\Hsuc(x,y)$&$:=$&$ P(x,y) \implies \bigvee\limits_{(t,t')\in H}\big( t(x)\land t'(y)\big)$\\
\hline &&\\
$\Vsuc(x,y)$&$:=$&$Q(x,y) \implies \bigvee\limits_{(t,t')\in V} \big( t(x) \land t'(y)\big)$
\\ \hline &&\\
$\children^0 := \top$& and & $\children^n = \Diamond \top \land \Box \big( \children^{n-1})$\\ 
\hline 
\end{tabular}
  \caption{Short hand formulas}
  \label{Fig-short formulas}
 \end{figure}

The formula $\onlyT(t,x)$ asserts that grid point $x$ is tiled with exactly $t$ and no other tile.  The formulas $\Hsuc(x,y)$ and $\Vsuc(x,y)$ ensure that the horizontal and vertical successors satisfy the tiling constraints respectively. Also,  $\children^n$ ensures that there is at least one path of length $n$ starting from the current world and all  paths starting from the current world can be extended to length at least $n$.

For better readability, we sometimes drop the brackets in the predicates $P(x,y)$ and $Q(x,y)$ and write $Pxy$ and $Qxy$ respectively. 
The tiling encoding formulas are described in Fig. \ref{Fig-BA+EB undec formulas}

\begin{figure}[h]
  \centering
  \footnotesize
  \begin{tabular}{|r l | }
\hline &\\
$\alpha_0:=$&$\Diamond \exists x_0~\Big\{ \Box \Box\big( \onlyT(t_0,x_0)\big)\Big\} \land \Box\forall x\Big\{ \bigvee\limits_{t\in T} \Box\Box \onlyT(t,x) \Big\} $\\
&$\land~ \children^3\top$\\
\hline &\\
$\alpha_H:=$&$\Box\forall x~\Big\{ \exists x_1 \Box \Big(\Box Pxx_1 \Big)\Big\}$
\\
$\alpha_V:=$&$\Box\forall x~\Big\{ \exists x_2 \Box \Big(\Box Qxx_2 \Big)\Big\}$\\ \hline &\\

$\alpha_{Hs}:=$&$\Box\forall x~\Big\{  \Box \forall y \Big(\Box \Hsuc(x,y) \Big) \Big\} $\\
$\alpha_{Vs}:=$&$\Box\forall x~\Big\{ \Box \forall y \Big(\Box \Vsuc(x,y) \Big) \Big\} $\\
\hline
&\\
$\phi_H:=$&$\Box\forall x~\Big\{   \Box \forall y\Big(  \Diamond Pxy \leftrightarrow \Box Pxy \Big) \Big\}$
\\
$\phi_V:=$&$\Box\forall x~\Big\{  \Box \forall y\Big(  \Diamond Qxy \leftrightarrow \Box Qxy \Big) \Big\}$\\
\hline &\\

$\psi:=$&$\Box\forall x~ \Big\{   \Box \forall y \Big( [\exists z'\Box (Qxz' \land P z'y) ]\implies [\Box \forall z (Pxz \implies Qzy)] \Big) \Big\}$\\
\hline
\end{tabular}

  \caption{$\EBBA$ formulas for encoding the tiling instance over increasing domain models}
  \label{Fig-BA+EB undec formulas}
\end{figure}

The formula $\alpha_0$ asserts that there is a grid point which is tiled with $t_0$ and every grid point has a unique tile (and $\children^3$ ensures that there are sufficiently many descendants). Further, $\alpha_H$ and $\alpha_V$ assert that every grid point  has a horizontal and vertical successor respectively. Formulas $\alpha_{Hs}$ and $\alpha_{Vs}$ ensure that the horizontal and vertical constraints of the input tiling instance are satisfied. 

The formulas $\phi_H$ and $\phi_V$ ensure that the horizontal and vertical successor information is uniform across all descendants respectively. Finally $\psi$ states the grid points satisfy the diagonal property.

\bigskip
For a given tiling instance $\T$,  the corresponding  formula $\alpha_T$ is given by the conjunction of the above formulas.

\begin{theorem}
\label{thm-EBBA undecidability}
For a given tiling instance $\T$, there exists a proper tiling of $\T$ over $\nat \times \nat$ iff $\alpha_T$ is satisfiable in an increasing domain model.
\end{theorem}
\begin{proof} $(\Rightarrow)$
Suppose $f: (\nat\times\nat) \to T$ is the tiling function. Then define the  model $\M = (\W,\D,\R,\val)$ where
\begin{itemize}
\item $\W = \{ w_0,w_1,w_2,w_3\}.$
\item $\D = \nat\times\nat$. \qquad \quad   
\item $\R = \{ (w_k,w_{k+1}) \mid 0\le k <3\}$
\item Define $\val$ as follows (for all $k\le 3$):\\
$\val(w_k,P) = \Big\{ \Big((i,j),~(i+1,j)\Big)\mid i,j\in \nat\Big\}$\\
$\val(w_k,Q) = \Big\{\Big((i,j),~(i,j+1)\Big) \mid i,j\in \nat\Big\}$\\
For every $t\in T$ define $\val(w_k,t) = \Big\{ (i,j) \mid f(i,j) = t\Big\}$
\end{itemize}
Valuation is important only at $w_3$ (valuations at other worlds are irrelevant). It can be verified that $\M,w_0\models \alpha_\T$.

\bigskip
$(\Leftarrow)$~ Suppose $\M,r\models \alpha_\T$. We fix an enumeration of the elements in the domain of $\M$ since it is countable (cf. footnote \ref{ft.LS}). Below, every time we pick a domain element satisfying a property, we pick the one that is least in the enumeration of domain elements satisfying that property. 

 Let $\alpha_0', \alpha_H',\alpha_V',\alpha_{Hs}',\alpha_{Vs}',\phi_H',\phi_V',\psi'$ be the respective formulas where the outermost modality is removed. Then by $\alpha_0$ there is some $r\to u$ such that\\  $\M,u\models \alpha_0'\land  \alpha_H'\land \alpha_V'\land \alpha_{Hs}'\land \alpha_{Vs}'\land \phi_H'\land \phi_V'\land \psi'$.

Define a mapping $g: (\nat \times \nat) \to \live(u)$ by induction as follows: let $d_0\in \live(u)$ such that $\M,u\models \Box\Box \onlyT(t_0,d_0)$. Define $g(0,0) = d_0$.
 
 For all $j > 0$ if $g(0,j-1) = c$ then let $d\in  \live(u)$ be a  witness for $x_2$ with the assignment $x\mapsto c$ in $\alpha_V$. So, we have $\M,u\models \Box\Box Qcd$. Define $g(0,j) = d$.
 
 For all $i,j > 0$ if $g(i-1,j) = c$ then let $d\in \live(u)$ be a  witness for $x_1$ with the assignment $x\mapsto c$ in $\alpha_H$. So, we have $\M,u\models \Box\Box Pcd$. Define $g(i,j) = d$.

Intuitively, we induce a grid over  $\live(u)$, by first building the $y$-axis of the first quadrant and then build horizontal lines fixing each $y$-coordinate in such a way that the vertically adjacent elements are also connected to satisfy the `grid property'.

Note that for all $i,j$ if $g(i,j) = c$ and $g(i+1,j) = c'$ then by definition, we have $\M,u\models \Box\Box Pcc'$. First we prove that if $g(i,j) = c$ and $g(i,j+1) = d$ then we  have $\M,u\models \Box\Box Qcd$. This is proved by induction on $i$.
 
 In the base case, $i=0$ and if $g(0,j) = c$ and $g(0,j+1) = d$ then  by construction we have $\M,u\models \Box\Box Qcd$.
 Now, for the induction step,  consider some $i>0$ and let $g(i,j)  = c$ and $g(i,j+1) = d$. Let $g(i-1,j) = a$ and $g(i-1,j+1) = b$.  By induction hypothesis, $\M,u\models \Box \Box Qab$. By construction we also have $\M,u\models \Box\Box Pac$ and $\M,u\models \Box\Box Pbd$.  \\
  We will prove that $\M,u\models \Box\Box Qcd$. For this, pick any successor $u\to v$. We will prove that $\M,v\models \Box Qcd$. Note that we have $M,v\models \Box Pac$ and $\M,v\models \Box Pbd$ and  $\M,v\models \Box Qab$.

 Now since $\M,v\models \children^1$, there is some $v\to w$  and  we have
  $\M,w\models Pac \land Pbd \land Qab$. Further,  by $\psi$ (assigning $x$ to $a$):  $\M,v\models \forall y \Big( [\exists z'\Box (Qaz' \land P z'y)] \implies [\Box \forall z (Paz \implies Qzy)] \Big)$. 
  By assigning $y$ to $d$ we have
  $\M,v\models [\exists z'\Box (Qaz' \land Pz'd)] \implies [\Box \forall z (Paz \implies Qzd) ]$.\\
   But note that we have $\M,v\models \Box\Big( Qab \land Pbd\Big)$. 
   Hence, $\M,v\models \Box \forall z (Paz \implies Qzd) $.
 
 Since $v\to w$, we have $\M,w\models \forall z (Paz\implies Qzd)$. Now, by assigning $z$ to $c$ we have $\M,w\models Pac$ and hence we have $\M,w\models Qcd$.  This implies that $\M,v\models \Diamond Qcd$. Now by $\phi_V$ we have $\M,v\models \Diamond Qcd \leftrightarrow \Box Qcd$. Hence $\M,v\models \Box Qcd$ as required.

\bigskip

The tiling function $f: (\nat\times \nat) \to T$ is defined as follows: For every $(i,j) \in (\nat\times \nat)$, suppose $g(i,j) = d$ then since $\M,u\models \forall x \Big\{\bigvee\limits_{t\in T} \Box\Box \onlyT(t,x) \Big\}$ there is some unique $t\in T$ such that $\M,u\models \Box \Box t(d)$. Define $f(i,j) = t$.

To prove that $f$ is a proper tiling function, first note that with $g(i,j) = d_0$ we have $f(0,0) = t_0$ since $\M,u\models \Box\Box \onlyT(t_0,d_0)$. Now pick some $i,j$ and let $g(i,j) = c$ and $f(i,j) = t$ . We verify that $f$ satisfies horizontal and vertical tiling constraints.

  Let $f(i+1,j) = t'$ and $g(i+1,j) = d$ be the horizontal successor. 
  By construction, $\M,u\models \Box\Box Pcd$ and $M,u\models \Box\Box \Big( t(c) \land t'(d) \Big)$. Since $\M,u\models \children^2$, there is some $u\to v\to w$, and hence we have $\M,w\models Pcd \land t(c) \land t'(d)$. But then by $\alpha_{Hs}$ we have $\M,u\models \Box\Box \Hsuc(c,d)$ and hence $(t,t')\in H$.

 Similarly for vertical constraints, suppose $g(i,j+1) = d$ and $f(i,j+1) = t'$ then note that we have already proved that in this case $\M,u\models \Box\Box Qcd$. Again, we  have $\M,u\models \Box\Box \Big( t(c) \land t'(d) \Big)$. Since $\M,u\models \children^2$, there is some $u\to v\to w$ and hence we have $\M,w\models Qcd \land t(c) \land t'(d)$. But then by $\alpha_{Vs}$ we have $\M,v\models \Box\Box \Vsuc(c,d)$ and hence $(t,t') \in V$.
\end{proof}

\begin{corollary}
Let $\mathcal{L}$ be any fragment of $\FOML$ such that $\EBBA\subseteq \mathcal{L}$. Then the satisfiability problem for $\mathcal{L}$ over increasing domain models is undecidable.
\end{corollary}

\subsection{$\ABEBBE$ over increasing domain}

In this fragment we have the bundled operators $\forall x \Box~\alpha$, $\exists x\Box~\alpha$ and $\Box\exists x~\alpha$ (and their duals). Hence, in this case also we can express $\forall x \exists y \Box \alpha$ and $\forall x \Box\forall y \Box\forall z \alpha$ as $\Diamond \forall x \exists y \Box \alpha$ and $\forall x \Box \forall y \Box \forall z \alpha$ respectively. Note that the former formula is a combination of $\BE$ and $\EB$ formulas and the later is an $\AB$ formula.

The tiling encoding formulas are described in Fig. \ref{Fig-AB+EB+BE undec formulas} which are the modification of the formulas in Fig. \ref{Fig-BA+EB undec formulas} to suit the $\ABEBBE$ fragment. 

\begin{figure}[h]
  \centering
  \footnotesize
 \begin{tabular}{|rrl | }
\hline
$\hat{\alpha_0}:=$&$\Box \exists x_0\Big\{$&$ \Box \Box\big( \onlyT(t_0,x_0)\big)\Big\} \land \children^3\top$\\
&&\\ \hline
&$\Diamond \forall x \Big\{$& \\
$\hat{\hat{\alpha}}_0:=$&&$ \bigvee\limits_{t\in T} \Box\Box \onlyT(t,x)~\land$ \\
$\hat{\alpha}_H:=$&&$ \exists x_1 \Box \Big(\Box Pxx_1 \big) ~\land$ \\
$\hat{\alpha}_V:=$&&$ \exists x_2 \Box \Big(\Box Qxx_2 \Big) ~\land$ \\

$\hat{\alpha}_{Hs} :=$&&$\forall y  \Box  \Big(\Box \Hsuc(x,y) \Big)~\land $ \\
$\hat{\alpha}_{Vs} :=$&&$  \forall y\Box \Big( \Box \Vsuc(x,y) \Big) ~\land$ \\

$\hat{\phi}_H:=$&&$   \forall y\Box \Big(  \Diamond Pxy \leftrightarrow \Box Pxy \Big)~\land $ \\
$\hat{\phi}_V:=$&&$ \forall y\Box \Big(  \Diamond Qxy \leftrightarrow \Box Qxy \Big) ~\land$ \\

$\hat{\psi}:=$&&$  \forall y \Box  \Big( [\exists z'\Box (Qxz' \land P z'y) ]\implies [\forall z\Box  (Pxz \implies Qzy)] \Big)$ \\
&$ \Big\}$& \\
\hline
\end{tabular}

  \caption{$\ABEBBE$ formulas to encoding the tiling instance over increasing domain models}
  \label{Fig-AB+EB+BE undec formulas}
\end{figure}

For a given tiling instance $\T$,  the corresponding  formula $\hat{\alpha}_T$ is given by the conjunction of the above formulas.

\begin{theorem}
\label{thm-ABEBBE undecidability}
The satisfiability problem for $\ABEBBE$ fragment over increasing domain models is undecidable.
\end{theorem}

\begin{proof}
 $(\Rightarrow)$ It can be verified that $\M,w_0 \models \hat{\alpha}_T$ where $\M$ is the model described in the proof of Theorem \ref{thm-EBBA undecidability}.

$(\Leftarrow)$   Suppose $M,r\models \hat{\alpha}_\T$. Then there is some $r\to u$ such that

 \begin{tabular}{rrl  }

&$ \exists x_0\Big\{$&$ \Box \Box\big( \onlyT(t_0,x_0)\big)\Big\} \land \children^2\top \land$\\

$\M,u\models$&$ \forall x \Big\{$& \\
&&$ \bigvee\limits_{t\in T} \Box\Box \onlyT(t,x)~\land$ \\
&&$ \exists x_1 \Box \Big(\Box Pxx_1 \big) ~\land$ \\
&&$ \exists x_2 \Box \Big(\Box Qxx_2 \Big) ~\land$ \\

&&$\forall y  \Box  \Big(\Box \Hsuc(x,y) \Big)~\land $ \\
&&$  \forall y\Box \Big(\Box \Vsuc(x,y) \Big) ~\land$ \\

&&$   \forall y\Box \Big(  \Diamond Pxy \leftrightarrow \Box Pxy \Big)~\land $ \\
&&$ \forall y\Box \Big(  \Diamond Qxy \leftrightarrow \Box Qxy \Big) ~\land$ \\

&&$  \forall y \Box  \Big( [\exists z'\Box (Qxz' \land P z'y) ]\implies [\forall z\Box  (Pxz \implies Qzy)] \Big)$ \\
&$ \Big\}$& \\

\end{tabular}

Now define a mapping $g: (\nat \times \nat) \to \live(u)$ by induction as follows: let $d_0\in \live(u)$
 such that $\M,u\models \Box\Box \onlyT(t_0,d_0)$. Define $g(0,0) = d_0$.
 
 For all $j > 0$ if $g(0,j-1) = c$ then let $d\in  \live(u)$ be a  witness for $x_2$ with the assignment $x\mapsto c$ in $\hat{\alpha}_V$. So, we have $\M,u\models \Box\Box Qcd$. Define $g(0,j) = d$.
 
 For all $i,j > 0$ if $g(i-1,j) = c$ then let $d\in \live(u)$ be a  witness for $x_1$ with the assignment $x\mapsto c$ in $\hat{\alpha}_H$. So, we have $\M,u\models \Box\Box Pcd$. Define $g(i,j) = d$.
 
\medskip
Note that for all $i,j$ if $g(i,j) = c$ and $g(i+1,j) = c'$ then we have $\M,u\models \Box\Box Pcc'$. Now we claim that if $g(i,j+1) = d$ then using the same reasoning as in the proof of Theorem \ref{thm-EBBA undecidability} we can argue that  $\M,u\models \Box\Box Qcd$. 
Now define the tiling function $f: (\nat\times \nat) \to T$ where for every $(i,j) \in (\nat\times \nat)$, suppose $g(i,j) = d$ then since $\M,u\models \forall x \Big\{\bigvee\limits_{t\in T} \Box\Box \onlyT(t,x) \Big\}$ there is some unique $t\in T$ such that $\M,u\models \Box \Box t(d)$. Define $f(i,j) = t$.  
 
Again, using the same arguments as in the proof of Theorem \ref{thm-EBBA undecidability}, we can verify that $f$ is indeed a proper tiling function.
\end{proof}

\section{Fragments without Finite Model Property}
\label{sec-noFmp}

In this section we consider the fragments that lack  Finite Model Property, or rather, that they can force models with infinite domains. Of course, this does not necessarily mean that the satisfiability problem for these fragments are undecidable but any strategy to prove decidability based on finite model property will fail. Abstractly speaking, these are the fragments in which we can express $\forall x \exists y~\Box\alpha$ but we cannot express $\forall x \Box\forall y \Box \forall z~\alpha$. 

The inability to express $\forall x \Box\forall y \Box \forall z~\alpha$ means that we cannot assert the `grid property' needed for the tiling encoding and consequently  we cannot prove undecidability.
On the other hand, using formulas of the form $\forall x \exists y~\Box \alpha$, we can induce an irreflexive-transitive ordering over the local domain at the world where the formula is satisfied. Then, stating that such an ordering does not have a maximal element will imply that the domain of the model has to be infinite. 

To encode the ordering we use a binary predicate $P$, with the intended meaning of $Pxy$ to mean $x<y$ in the ordering being defined.

\subsection{$\EBBE$ over increasing domain}

In this fragment we are allowed $\exists x\Box~\alpha$ and $\Box\exists x~\alpha$ (and their duals) in the syntax. Hence we can assert $\forall x\exists y~\Box \alpha$ in the form of $\Diamond \forall x\exists y\Box~\alpha$. But we cannot express $\forall x\Box\forall y\Box \forall z\alpha$. This is because every $\forall$ operator needs to necessarily have a $\Diamond$ either before or after its occurrence. 

We prove that this fragment lacks the finite model property.
The encoding formula is given as follows:

\begin{tabular}{l r l l l}
$\phi_1:=$&$  \Diamond\forall x \Big[ $&$~\exists y \Box\Box Pxy~\land \Box\Box \neg Pxx ~\land$\\ 
&&$\Diamond \forall y\Big(~\big[\Diamond Pxy\leftrightarrow \Box Pxy~\big]~\land$\\
&&$\quad\quad\quad\quad \Diamond \forall z\big[ \big(Pxy \land Pyz\big) \implies \big(Pxz\big) \big]\Big)~\Big]$
\end{tabular}


The first conjunct asserts that every $x$ has a successor $y$ and the second conjunct asserts that the binary predicate $P$ is irreflexive. The last conjunct asserts that $P$ is transitive and the one before ensures that the relation $Pxy$ holds uniformly across all branches.

\begin{theorem}
\label{thm-EBBE has no FMP}
The formula $\phi_1$ is satisfiable in a model with infinite $\D$. Moreover, for all increasing domain model $\M$ and $w\in \W$ if $\M,w\models \phi_1$ then there is some $w\to u$ such that $\live(u)$ is infinite.
\end{theorem}
\begin{proof}
Consider the model $\M=(\W,\nat,\R,\live,\val)$ where:\\
 $\W = \{w_0,w_1,w_2\}$ and $\R = \{ (w_j,w_{j+1}) \mid j<2\}$. For all $j\le 2$ define $\live(w_j) = \nat$ and $\val(w_j,P) = \{ (k,l) \mid k,l\in \nat$ and $k<l\}$.
Note that only $\val(w_2,P)$ is relevant. It can be verified that $\M,w_0\models \phi_1$. 

Now assume that $\M,r\models \phi_1$. Then  there exists some $r\to u$ such that

\begin{tabular}{l r l l l}
$\M,u\models$&$ \forall x \Big[ $&$~\exists y \Box\Box Pxy~\land \Box\Box \neg Pxx ~\land$\\ 
&&$\Diamond \forall y\Big(~\big[\Diamond Pxy\leftrightarrow \Box Pxy~\big]~\land$\\
&&$\quad\quad\quad\quad \Diamond \forall z\big[ \big(Pxy \land Pyz\big) \implies \big(Pxz\big) \big]\Big)~\Big]$
\end{tabular}


We will prove that $\live(u)$ is infinite. Towards this we will construct an infinite sequence of distinct domain elements $d_0,d_1\ldots$ over $\delta(u)$. 
 The sequence is constructed by induction such that  for all $i$:
 $\quad\M,u,[x\mapsto d_i,y\mapsto d_{i+1}]\models \Box\Box Pxy$
 
 In the base case let since $\live(u)$ is non-empty, pick some arbitrary $d_0\in \live(u)$.\\
 For the induction step, suppose we have constructed $d_0,\ldots d_i$. 
   Let $d_{i+1}$ be such that $\M,u,[x\mapsto d_i,y\mapsto d_{i+1}]\models \Box\Box Pxy$ (such an element exists since $\M,u\models \forall x \exists y \Box\Box Pxy$).
 
 We claim that $d_{i+1}$ is distinct from $d_0\ldots d_i$. Suppose not; then let $d_{i+1} = d_k$ for some $k \le i$. Let $k+l = i$ for some $l\ge 0$.  
 
Now, with the assignment $x \mapsto  d_{i}$, we have

\begin{tabular}{l  l l l}
$\M,u\models$&$\Diamond \forall y\Big(~\big[\Diamond Pd_iy\leftrightarrow \Box Pd_iy~\big]~\land$\\
&$\quad\quad\quad \Diamond \forall z\big[ \big(Pd_iy \land Pyz\big) \implies \big(Pd_iz\big) \big]\Big)$
\end{tabular}

Hence there is some $u\to w$ such that

\begin{tabular}{l  l l l}
$\M,w\models $&$ \forall y\Big(~\big[\Diamond Pd_iy\leftrightarrow \Box Pd_iy~\big]~\land$\\
&$\quad\quad\Diamond \forall z\big[ \big(Pd_iy \land Pyz\big) \implies \big(Pd_iz\big) \big]\Big)$
\end{tabular}

  By induction hypothesis,  $\M,w\models \Box Pd_jd_{j+1}$ for all $j\le i$ and by assumption we have $\M,w\models \Box Pd_id_k$  (since $d_{i+1} = d_k$). 
  
  Now with the assignment $y\mapsto d_k$ there is some $w\to w^k$ such that
  $\M,w^{k} \models \forall z~( Pd_id_k \land Pd_k z \implies Pd_iz)$. In particular with $z \mapsto d_{k+1}$, we obtain $\M,w^{k} \models Pd_id_{k+1}$.  
  This implies that $\M,w\models \Diamond Pd_id_{k+1}$.     
  
  Now since $\M,w \models \forall y \big(\Diamond Pd_iy\leftrightarrow \Box Pd_iy\big)$, we have $\M,w \models \Box Pd_id_{k+1}$.     
 But again, with  $y\mapsto d_{k+1}$ there is some  $w\to w^{k+1}$ be such that $\M,w^{k+1} \models \forall z~( Pd_id_{k+1} \land Pd_{k+1}z \implies Pd_iz~)$. In particular with $z \mapsto d_{k+2}$, we have $\M,w^{(k+1)} \models Pd_id_{k+2}$. This implies that $\M,w\models \Diamond Pd_id_{k+2}$ from which we can conclude  that $\M,w\models \Box Pd_id_{k+2}$.
  
 Applying this argument $l$ times, we get $\M,w\models \Diamond Pd_id_{k+l}$. Hence we have $\M,w\models \Diamond Pd_id_i$ (since $k+l = i$) which is a contradiction to $\M,u\models \forall x~ \Box\Box \neg Pxx$.
\end{proof}

\subsection{$\forall\exists\Box$ fragment over increasing domain}

Consider any fragment that can express $\forall x\exists y \Box \alpha$ formulas. We  prove that such a fragment does not have finite model property. For instance, any extension of the quantifier prefix of $\LBF$ fragment will allow us to express   $\forall x\exists y\Box~\alpha$ formulas. Similarly the negation closed extension of $\LBF$ also allows us to express $\forall x\exists y\Box~\alpha$.

Note that the ability to express $\forall x\exists y\Box~\alpha$ implies that we can also write formulas of the form $\forall x \Box\alpha$ (by making $y$ as a dummy variable) and also $\Box\alpha$ (by making both $x,y$ as dummy variables). The encoding formula is given as follows:

\begin{tabular}{l l}
$\phi_2~:=$&$ \forall x \exists y \Box\big( \Box \Box Pxy\big) ~\land~ \forall x \Box \big(\Box\Box \neg Pxx\big) ~ \land$\\
&$ \forall x \Box \forall y \Box\forall z\Box \big(Pxy \land Pyz \implies Pxz \big) \land \Diamond\Diamond\Diamond \top$
\end{tabular}

The formula $\phi_2$ is essentially the same as $\phi_1$ only modified to ensure that it can be written using only $\forall x\exists y \Box\alpha$ formulas.

\begin{theorem}
\label{thm-AEB has no FMP}
The formula $\phi_2$ is satisfiable in a model with infinite $\D$. Moreover, for all increasing domain models $\M$ and $w\in \W$, if $\M,w\models \phi_2$ then $\live(w)$ is infinite.
\end{theorem}

\begin{proof}
Consider the model $\M=(\W,\nat,\R,\live,\val)$ where:\\
 $\W = \{w_0,w_1,w_2,w_3\}$ and $\R = \{ (w_j,w_{j+1}) \mid j<3\}$. For all $j\le 3$ define $\live(w_j) = \nat$ and $\val(w_j,P) = \{ (k,l) \mid k,l\in \nat$ and $k<l\}$.
Note that only $\val(w_3,P)$ is relevant. It is easy to verify that $\M,w_0\models \phi_2$.

Now assume that $\M,w\models \phi_2$. We will construct an infinite sequence of distinct domain elements $d_0,d_1\ldots$ over $\live(w)$ and this proves that $\live(w)$ is infinite. 
 The sequence is constructed by induction. We maintain the invariant that for all $i$ we have $\M,w,[x\mapsto d_i, y\mapsto d_{i+1}]\models \Box\Box\Box Pxy$.

 In the base case, since $\live(w)$ is non-empty, pick some arbitrary $d_0\in \live(w)$.
 For the induction step, suppose we have constructed $d_0,\ldots d_i$. 
   Let $d_{i+1}$ be such that $\M,w,[x\mapsto d_i, y\mapsto d_{i+1}]\models \Box\Box\Box Pxy$ \Big(such an element always exists since $\M,w\models \forall x \exists y \Box\Box\Box Pxy$\Big).

  Now we claim that $d_{i+1}$ is distinct from $d_0\ldots d_i$. Suppose not then let $d_{i+1} = d_j$ for some $j \le i$. Let $j+l = i$ for some $l\ge 0$.  Since $\M,w\models \Diamond\Diamond\Diamond\top$, there exists $w\to w'\to u\to v$. Also, since $\M,w\models \forall x \Box \forall y \Box \forall z \Big( Pxy\land Pyx \implies Pxz)$, it follows that for all $a,b,c\in \live(w)$ we have $\M,v\models (Pab\land Pbc) \implies Pac$.
  
  From the construction of the sequence, we have $\M,v\models Pd_id_{i+1} \land Pd_jd_{j+1}$ and hence we have $\M,v\models Pd_id_{j+1}$. But then again since $\M,v,\models Pd_{j+1}d_{j+2}$ we have $Pd_id_{j+2}$. Repeating this for $l$ times, we get $\M,v\models Pd_id_{i}$ but this is a contradiction to $\M,w \models \forall x\Big(\Box\Box\Box \neg Pxx\Big)$.

\end{proof}

\subsection{$\BE$ over constant domain}

When we consider  constant domain models, we can relax the condition that $\forall x\exists y \Box\alpha$ quantification has to be over the same domain. Since the local domain at every world is the same, all quantifiers at all worlds are evaluated over the same domain.  Hence, in $\BE$, using the $\Box\exists x~\alpha$ operator we can write formula of the form $\Diamond \forall x \Box \exists \alpha$.  Now we prove that this is sufficient to prove the lack of finite model property. The encoding formulas are given as follows:

{\small
\begin{tabular}{l r l l l}
$\phi_3:=$&$  \Diamond\forall x \Big[ $&$ \Box\exists y~\Box\Box Pxy~\land\Box\Box\Box \neg Pxx ~\land$\\ 
&&$\Diamond\Big\{\Diamond \forall y\Big(~\big[\Diamond Pxy\leftrightarrow \Box Pxy~\big]~\land$\\
&&$\quad\quad\quad\quad \Diamond \forall z\big[ \big(Pxy \land Pyz\big) \implies \big(Pxz\big) \big]\Big)~\Big\}~\Big]$
\end{tabular}
}

The formula is similar to the formula $\phi_1$, except that we need to insert one more $\Box$ as the padding due to the fact that we can only introduce an existential quantifier after a $\Box$ in $\BE$.

\begin{theorem}
\label{thm-BE over constant domain has no FMP}
The formula $\phi_3$ is satisfiable in a constant domain model with infinite $\D$. Moreover, for all constant domain model $\M$ and $w\in \W^M$ if $\M,w\models \phi_3$ then $\D$ is infinite.
\end{theorem}

\begin{proof}
We can verify that $\M,w_0\models \phi_3$ where $\M$ is the model described in the proof of Theorem \ref{thm-AEB has no FMP} (Note that $\M$ is a constant domain model).

  For the moreover part, let $\M$ be an arbitrary constant domain model such that $\M,r\models \phi_3$. Then we will prove that $\D$ is infinite. 
First note that there exists some $r\to r'$ such that\\
{\small
\begin{tabular}{l r l l l}
$\M,r'\models $&$  \forall x \Big[ $&$ \Box\exists y~\Box\Box Pxy~\land\Box\Box\Box \neg Pxx ~\land$\\ 
&&$\Diamond\Big\{\Diamond \forall y\Big(~\big[\Diamond Pxy\leftrightarrow \Box Pxy~\big]~\land$\\
&&$\quad\quad\quad\quad \Diamond \forall z\big[ \big(Pxy \land Pyz\big) \implies \big(Pxz\big) \big]\Big)~\Big\}~\Big]$
\end{tabular}
}

Since $\M$ is a constant domain model, here is some successor $r'\to u$ such that:\\

\begin{tabular}{l r l l l}
$\M,u\models$&$  \forall x \Big[ $&$\exists y~\Box\Box Pxy~\land\Box\Box \neg Pxx ~\land$\\ 
&&$\Big\{\Diamond \forall y\Big(~\big[\Diamond Pxy\leftrightarrow \Box Pxy~\big]~\land$\\
&&$\quad\quad\quad\quad \Diamond \forall z\big[ \big(Pxy \land Pyz\big) \implies \big(Pxz\big) \big]\Big)~\Big\}~\Big]$
\end{tabular}

\bigskip

 We will construct an infinite sequence of distinct domain elements $d_0,d_1\ldots$ over $\D$ and this will prove that $\D$ is infinite. 
 The sequence is constructed by induction. We maintain that the invariant that for all $i$: $\M,u \models \Box\Box Pd_id_{i+1}$.
 
 \medskip
 
 In the base case, since $\D$ is non-empty, pick some arbitrary $d_0\in \D$.
 For the induction step, suppose we have constructed $d_0,\ldots d_i$. 
   Let $d_{i+1}$ be such that $\M,u\models \Box\Box Pd_id_{i+1}$ (such an element always exists since $\M,u\models \forall x \exists y \Box\Box Pxy$).

  Now we claim that $d_{i+1}$ is distinct from $d_0\ldots d_i$. Suppose not then let $d_{i+1} = d_j$ for some $j \le i$. Let $j+l = i$ for some $l\ge 0$.  
  Let $u\to v^i$ such that with $x\mapsto d_i$ we have\\
  
$\M,v^i\models \forall y\Big(~\big[\Diamond Pd_iy\leftrightarrow \Box Pd_iy~\big]~\land \Diamond \forall z\big[ \big(Pd_iy \land Pyz\big) \implies Pd_iz \big]\Big)$

  By construction $\M,v^i\models \Box Pd_kd_{k+1}$ for all $k\le i$. Also, by assumption that $d_{i+1} = d_j$, we have $\M,v^i\models \Box Pd_id_j$.  We also have $\M,v^i\models \Box\Box\neg Pd_id_i$.

  Let $v^i\to w^{ij}$ be such that with  the assignment $y\mapsto d_j$ we have
  $\M,w^{ij}\models \forall z~( Pd_id_j \land Pd_jz \implies Pd_iz~)$. In particular with $z \mapsto d_{j+1}$, we obtain $\M,w^{ij} \models Pd_id_{j+1}$.  
  This implies that $\M,v^i\models \Diamond Pd_id_{j+1}$. \\  
  Since $\M,v^i \models \forall y~\big(~\big(\Diamond Pd_iy\leftrightarrow \Box Pd_iy~\big)$ we have $\M,v^i \models \Box Pd_id_{j+1}$. 
  
  Now, let $v^i\to w^{i(j+1)}$ be such that $\M,w^{i(j+1)}\models \forall z~( Pd_id_{j+1} \land Pd_{j+1}z \implies Pd_jz~)$. In particular with $z \mapsto d_{j+2}$, we have $\M,w^{i(j+1)} \models Pd_id_{j+2}$. This implies that $\M,v^i\models \Diamond Pd_id_{j+2}$ from this we can conclude  that $\M,v^i\models \Box Pd_id_{j+2}$.\\
  
 Applying this argument $l$ times, we obtain that $\M,v^i\models \Diamond Pd_id_i$ and hence we obtain a contradiction.

\end{proof}

\section{Decidable fragments}

Over constant domains, it is proved in \cite{padmanabha_et_al:LIPIcs:2018:9942} that the $\EB$ fragment is decidable and $\AB$ is undecidable. In \cite{Liu2019bundled} the undecidability of $\texttt{BA}$ over constant domain is proved. Further, in Section \ref{sec-noFmp} we have proved that the $\BE$ (and hence $\BABE$ ) fragments do not have the finite model property. This completes the picture of bundled fragments  over constant domain models.

Over increasing domain models, in \cite{padmanabha_et_al:LIPIcs:2018:9942, Liu2019bundled},  the $\ABEB$ fragment and $\BABE$ fragments are respectively  shown to be decidable. In Section \ref{sec-undec} we have proved that the fragments $\EBBA$ and $\ABEBBE$ are undecidable and $\EBBE$ lacks the finite model property. The remaining cases are the fragments \LBF and \ABBABE. In this section we take up these two fragments and prove them to be decidable. Note that in both of these fragments the formula $\forall x \exists y \Box\alpha$ is not (syntactically) expressible.

We assume that the formulas are in negation normal form (where $\neg$ appears only in front of atomic predicates). We first define some useful terms and notations.

\begin{definition}
\label{def-module }
For any $\FOML$ formula $\phi$:
\begin{itemize}
\item $\phi$ is a \literal if $\phi$ is of the form $P(x_1,\ldots x_n)$ or of the form $\neg P(x_1,\ldots x_n)$
\item $\phi$ is a \module if $\phi$ is a \literal or $\phi$ is of the form $\Delta \alpha$ where $\Delta \in \{\Box,\Diamond\}$
\item The component of $\phi$ is defined inductively as follows:
\begin{itemize}
 \item If $\phi$ is a \module then $\comp(\phi) = \{\phi\}$
 \item If $\phi$ is of the form $\phi_1\land \phi_2$ or $\phi_1\lor \phi_2$ then\\ $\comp(\phi) = \comp(\phi_1) \cup \comp(\phi_2)$
 \item If $\phi$ is of the form $\forall x~ \phi_1$ or $\exists x~\phi_1$ then\\ $\comp(\phi) = \{ \phi\} \cup \comp(\phi_1)$ 
\end{itemize}

 \item A formula $\phi$ is called \Esafe if every  $\psi\in \comp(\phi)$ is a module or  of the form $\forall x~ \psi'$. A finite set of formulas $\Gamma$ is \Esafe if every  $\phi\in \Gamma$ is \Esafe.
 \end{itemize}
\end{definition}
Intuitively, $\comp(\phi)$ is the set of all subformulas of $\phi$ that are `to be evaluated' at the current world. An existential-safe formula $\phi$ does need witnesses from the current local domain in order to make the formula true. The notions of components and existential-safeness will play a role in the  tableau-based decision procedure to be introduced below. 

Before going into the specific tableau rules for various bundled fragments, we first explain the general method. A tableau is a tree-like structure generated by repeatedly applying a few rules from a single formula $\alpha$ with some auxiliary information as the root of the tree. Intuitively, a tableau for $\alpha$ is a pseudo model which can be transformed into a real model of $\alpha$ under some simple consistency conditions. We can then decide the satisfiability of a formula by trying to find a proper tableau. As in \cite{Wang17d}, a tableau $T$ in our setting is a tree structure such that each node is a triple $(w,\Gamma,\sigma)$ where $w$ is a symbol or a finite sequence of symbols intended as the \textit{name} of a possible world in the real model, $\Gamma$ is a finite set of  $\FOML$-formulas, and $\sigma$ is an assignment function for variables. Since we intend to use the set of variables as the domain in the tableau-induced real model, $\sigma$ is simply a \textit{partial} identity function on $\Var$, i.e., $\sigma(x)=x$ for all $x\in Dom(\sigma)\subseteq \Var$, where the domain of $\sigma$, $Dom(\sigma)$, is intended to be the local domain of the real model. The intended meaning of the node $(w,\Gamma,\sigma)$ is that all the formulas in $\Gamma$ are satisfied on $w$ with the assignment $\sigma$, thus we also write $(w: \Gamma, \sigma)$ for the triple. 

A tableau rule specifies how the node in the premise of the rule is transformed to or connected with one or more new nodes given by the conclusion of the rule.   Applying the rules can generate a tree-like structure, a tableau, which is saturated if every leaf node contains only literals. For any formula $\phi$, we refer to a saturated tableau of $\phi$ simply as a tableau of $\phi$. Further, a saturated tableau is \textit{open} if in every node $(w: \Gamma,\sigma)$ of the  tableau, $\Gamma$ does not contain both $\alpha$ and $\neg \alpha$ for any formula $\alpha$.\footnote{Refer \cite{Wang17d} for an illustration of a similar tableau construction.}

We call a formula {\em clean} if no variable occurs both bound and free in it and
every use of a quantifier quantifies a distinct variable. A finite set of formulas $\Gamma$ is {\em clean} if $\bigwedge\Gamma$, the conjunction of all formulas in $\Gamma$, is clean. Note that every 
$\FOML$-formula can be rewritten into an equivalent clean formula.  For instance, the formulas
$\exists x \Box Px \lor \forall x \Diamond Qx$ and $Px \land  \Box \exists x Qx$  are not clean, whereas $\exists x \Box Px \lor \forall y \Diamond Qy$ and 
$Px \land  \Box \exists y Qy$ are their clean equivalents respectively.
Clean formulas help in handling the witnesses for existential formulas in the tableau in a syntactic way.

Consider a finite set of formulas $\Gamma$ that is clean. Suppose we want to expand $\Gamma$ to $\Gamma \cup \{\alpha_1,\ldots \alpha_k\}$,  then even if each of $\alpha_i$ is clean, it is possible that a bounded variable of $\alpha_i$  also occurs in some $\phi \in \Gamma$ or another $\alpha_j$. To avoid this, first we rewrite the bound variables in each  $\alpha_i$ one by one by using the fresh variables that do not occur in $\Gamma$ and other previously rewritten $\alpha_j$.
Such a rewriting can be fixed by always using the first fresh variable in a fixed  enumeration of all the variables. When $\Gamma$ and $\{\alpha_1,\ldots \alpha_k\}$ are clear from the context, we denote $\alpha_i^*$ to be such a fixed rewriting of $\alpha_i$ into a clean formula. It is not hard to see that the resulting finite set $\Gamma \cup \{\alpha_1^*,\ldots \alpha_k^*\}$ is clean.

\subsection{$\LBF$ over increasing domain}

 \begin{figure}
 \begin{center}
 \begin{tabular}{|c c|}
\hline
\multicolumn{2}{|c|}{} \\
     $\dfrac{w:\phi_1\lor\phi_2,\Gamma,\sigma}{w: \phi_1,\Gamma,\sigma\mid w:\phi_2,\Gamma,\sigma}\ (\lor)$ &  $\dfrac{w:\phi_1\land\phi_2,\Gamma,\sigma}{w:\phi_1,\phi_2,\Gamma,\sigma} \ (\land)$ \\
    \multicolumn{2}{|c|}{} \\
     \hline
     \hline
      
      \multicolumn{2}{|c|}{}\\
      \multicolumn{2}{|c|}{$\dfrac{w:\exists x \phi,~ \Gamma,~\sigma}
       {w: \phi,~ \Gamma,~\sigma'}(\exists)$} \\& \\
        \multicolumn{2}{|c|}{ where $\sigma' = \sigma \cup \{(x,x) \}$}\\
      \multicolumn{2}{|c|}{} \\ 
      \hline
      \hline
      
       \multicolumn{2}{|c|}{}\\
      \multicolumn{2}{|c|}{$\dfrac{w:\forall y\phi,~ \Gamma,~\sigma}
       {w:~\{\phi^*[z/y] \mid z \in Dom(\sigma)\},~ \Gamma,~\sigma}(\forall)$} \\& \\
        \multicolumn{2}{|c|}{ where $\Gamma$ is \Esafe and}\\
        \multicolumn{2}{|c|}{every $\phi^*[z/y]$ is a clean rewriting of $\phi[z/y]$}\\
        \multicolumn{2}{|c|}{ with respect to $\Gamma \cup \{\phi[z/y]\mid z\in Dom(\sigma) \}$}\\
      \multicolumn{2}{|c|}{} \\ 
      \hline
      \hline
      
       \multicolumn{2}{|c|}{Given $n\geq 1$: and $m,s\geq 0$} \\
      \multicolumn{2}{|c|}{}\\
       \multicolumn{2}{|c|}{$\cfrac{\genfrac{}{}{0pt}{0}{w:\Diamond \phi_1,\ldots,\Diamond \phi_n}{\Box \beta_1,\ldots \Box \beta_m} \ l_1,\cdots,l_s,~~\sigma}{\left\langle wv_i: \phi_i,~\{\beta_j \mid j\in [1,m]\},~ \sigma\right\rangle \ \text{for all} \ i\in [1,n] } \quad (\Diamond)$} \\ 
       \multicolumn{2}{|c|}{} \\ 
      \hline
      \hline
      
      \multicolumn{2}{|c|}{Given $m \geq 1, s\geq 0$:} \\
      \multicolumn{2}{|c|}{} \\
      \multicolumn{2}{|c|}{$\dfrac{w:\Box\beta_1,\ldots \Box \beta_m,~l_1,\ldots,l_s,\sigma}{w:l_1,\cdots,l_s,\sigma} \quad (\END)$} \\
      \multicolumn{2}{|c|}{} \\
      \hline
\end{tabular}

\caption{Tableau rules for $\LBF$, here every $l_i$ is a \literal}
\label{fig-tableau for LBF}
\end{center}
\end{figure}

Tableau rules for  $\LBF$ fragment are described in Fig. \ref{fig-tableau for LBF}. 
 The $(\land)$ and $(\lor)$ rules are standard, where we make a non-deterministic choice of one of the branches for $(\lor)$. The rule $(\END)$ says that  if we are left with only modules and there are no $ \Diamond$
formulas, then the branch does not need to be explored further. The $(\Diamond)$ rule creates one successor world for every $\Diamond$ formula at the current node and includes all the $\Box$ formulas that need to be satisfied along with the $\Diamond$ formula and $\sigma$ is inherited in the successor worlds to preserve increasing domain property. The $(\exists)$ rule picks $x$ itself as a witness to satisfy $\exists x \phi$ and $(\forall)$ rule expands the set of formulas to include a clean version of $\phi[z/y]$ for every variable $z$ in the current local domain.

Note that only the $(\Diamond)$ rule can change (the name of) the possible world thus creating a new successor. It simply extends the name $w$ by new symbols $v_i$ for each successor. Therefore there can be many  nodes in the tableau sharing the same world name but such nodes form a path. Given $w$ we use $t_w$ to denote the last node sharing the first component $w$. Given a node $t=(w: \Gamma, \sigma)$ in a tableau, we use $Dom(t)$ to denote the domain of the domain of $\sigma$.

Also, there is an implicit ordering on how rules
are applied: $(\Diamond)$ rule can be applied at a node $(w,\Gamma,\sigma)$ only if all formulas of $\Gamma$ are modules
and hence may be applied only after the $(\land,\lor, \forall, \exists)$ rules have been applied as 
many times as necessary at $w$. Similarly $(\forall)$ rule can be applied only when $\Gamma$ is \Esafe which means that the $(\exists)$ rule cannot be applied any more at the current node.

\begin{proposition}
\label{prop-some rule can be applied always}
For every tableau $T$ and every node $v= (w,\Gamma,\sigma)$ in $T$ if $v$ is a leaf then either $\Gamma$ contains only literals or there is some rule that can be applied at $v$.
\end{proposition}

\begin{proof}
Suppose $\Gamma$ contains at least one non-literal formula. If there is some $\phi\in \Gamma$ where $\phi$ is  of the form $\phi_1\land \phi_2$ or $\phi_1 \lor \phi_2$ then we can apply $(\land)$ or $(\lor)$ rule respectively.  So assume that the operators $\land, \lor$ does not occur as the outer most connective in any formula of $\Gamma$.

Further, if every $\phi\in \Gamma$ is a \module then we can apply $(\Diamond)$ rule if 
at least one formula of the form $\Diamond \phi_1\in \Gamma$, otherwise we can apply (\END) rule.

So the remaining case is that there is at least one formula in $\Gamma$ which has a quantifier at the outermost level. Now if we have some $\exists x \phi_1\in \Gamma$ then we can apply $(\exists)$ rule. Otherwise every quantified formula in $\Gamma$ is of the form $\forall y \phi$. But then, from the syntax of $\LBF$ it follows that $\Gamma$ is \Esafe and hence we can pick some $\forall y \phi \in \Gamma$ and apply $(\forall)$ rule.
\end{proof}

\begin{theorem}
\label{thm-tableau for LBF}
For any clean $\LBF$ formula $\theta$, let $Dom(\sigma_r)= \FV(\theta)\cup\{z\}$ where $z$ does not occur in $\theta$. There is an open tableau $T$ with root $(r:\{\theta\},\sigma_r)$ iff $\theta$ is satisfiable in an increasing domain model. 
\end{theorem}

\begin{proof}
First we claim that the rules preserve cleanliness of the formulas. To see this, we verify that for every rule, if $\Gamma$ in the antecedent of the rule is clean then the $\Gamma'$ obtained after the application of the rules is also clean. This is obvious for $(\land), (\lor), (\Diamond)$ and (\END) rules. The $(\exists)$ rule preserves cleanliness because it frees variable $x$ which are not bound by any  other quantifier in the antecedent. The $(\forall)$ rule preserves cleanliness by rewriting. 

($\Rightarrow$): Let $T$ be an open tableau rooted at $(r:\{\theta\},\sigma_r)$.  Define a model $\M=(\W,\D,\live,\R,\val)$ as follows: 
\begin{itemize}
    \item $\W=\{w\mid (w: \Gamma,\sigma)$ is a node in $T\}$
    \item $\D=\Var$
    \item $\R = \{ (w,v) \mid v$ is of the form $wv'$ for some $v'\}$
    \item For every $w\in \W$ define $\live(w)=Dom(t_w)$\\
     where $t_w$ is the last node of $w$ in $T$ 
    \item For every $w\in \W$ and $p\in \Ps$ define \\
    $\val(w,P) = \{ \overline{x} \mid P\overline{x} \in \Gamma$ where $t_w = (w,\Gamma,\sigma)\}$
\end{itemize}

Clearly, $\mathcal{M}$ is an increasing domain model, and since $z\in Dom(\sigma_r)$, there is no empty local domain. As $T$ is an open tableau, $\val$ is well-defined. 

{\it Claim. } For every node $(w:\Gamma,\sigma)$ in $T$  and for every $\LBF$ formula $\phi$ if $\phi\in \Gamma$ then  $\M,w,\sigma\models \phi$.

\medskip
The proof of the claim is by induction on the nodes of $T$ from leaf nodes to the root. For the base case, $(w:\Gamma,\sigma)$ is a leaf node and hence $\Gamma$ contains only literals. Thus, by the definition of $\val(w)$, the claim holds as $\sigma$ is an identity assignment. 

For inductive step, $(w:\Gamma,\sigma)$ is non leaf node and hence some rule is  applied at this node and we have one or more descendants depending on the rule. Also, the claim holds for all the descendants. Now we consider various cases depending on which rule was applied.

\begin{itemize}
\item If $(\land)$ rule was applied then $\Gamma$ is of the form $\Gamma'\cup \{\phi_1\land \phi_2\}$ and the node $(w:\Gamma,\sigma)$ has one descendant $(w: \Gamma'\cup\{\phi_1,\phi_2\},\sigma)$. By induction hypothesis $\M,w,\sigma \models \phi_1$ and $\M,w,\sigma \models \phi_2$. Thus we have $\M,w,\sigma \models \phi_1\land \phi_2$. Further, by induction for every $\phi'\in \Gamma'$ we have $\M,w,\sigma\models \phi'$ and hence the claim holds.

\item If $(\lor)$ rule was applied then $\Gamma$ is of the form $\Gamma'\cup \{\phi_1\lor \phi_2\}$ and the node $(w:\Gamma,\sigma)$ has one descendant $(w: \Gamma'\cup\{\phi_i\},\sigma)$ for some $i\in \{1,2\}$. By induction hypothesis $\M,w,\sigma \models \phi_i$. Thus we have $\M,w,\sigma \models \phi_1\lor \phi_2$. Further, by induction for every $\phi'\in \Gamma'$ we have $\M,w,\sigma\models \phi'$ and hence the claim holds.

\item If $(\exists)$ rule was applied then $\Gamma$ is of the form $\Gamma'\cup \{\exists x \phi_1\}$ and the node $(w:\Gamma,\sigma)$ has one descendant $(w:\Gamma'\cup \{\phi_1\},\sigma')$ where $\sigma' = \sigma\cup \{(x,x)\}$.

By Induction hypothesis, $\M,w,\sigma'\models\phi_1$ and since $\sigma'=\sigma\cup\{(x,x)\}$, it follows that $\M,w,\sigma_{[x\mapsto x]}\models \exists x\phi_1$. Also since $\Gamma$ is clean, $x$ does not occur in any $\phi'\in \Gamma'$. Thus for every $\phi'\in \Gamma'$ by induction hypothesis we have $\M,w,\sigma'\models \phi'$ and hence $\M,w,\sigma \models \phi'$.

\item If $(\forall)$ rule is applied at $(w: \Gamma,\sigma)$ then $\Gamma$ is of the form $\Gamma'\cup \{\forall y \phi\}$ and we have a descendant $(w: \Gamma' \cup \{\phi^*[z/y]\mid z\in Dom(\sigma)\}, \sigma)$.

Also note that since we applied $(\forall)$ rule, this means $\Gamma$ is \Esafe. Hence, by syntax of $\LBF$, the $(\exists)$-rule is not applied at any of the descendants of $(w:\Gamma,\sigma)$ in the tableau where the nodes are of the form $(w: \Gamma_1,\sigma_1)$. This implies that $Dom(\sigma)=Dom (t_w)=\live(w)$. 

Thus by induction hypothesis,  for every $z\in \live(w)$ we have $\M,w,\sigma \models \phi*[z/y]$. But then $\phi^*[z/y]$ is equivalent to $\phi^*$, so it follows that $\M,w,\sigma_{[y\mapsto z]}\models\phi$.  Again, by induction hypothesis, for every $\phi'\in \Gamma'$ we have $\M,w,\sigma \models \phi'$ and hence the claim holds. 

\item If $(\Diamond)$ rule is applied at $(w:\Gamma,\sigma)$ then $\Gamma$ is of the form $\{l_1\ldots l_s\} \cup \{\Diamond\phi_1\ldots \Diamond\phi_n\}\cup \{\Box\beta_1\ldots \Box\beta_m\}$ for some $s,m\ge 0$ and $n\ge 1$.

Consequently there are $n$ children each of the from $(wv^i: \Gamma^i,\sigma)$ where $\Gamma^i = \{\phi_i\} \cup \{ \beta_1\ldots \beta_m\}$. Further, we have $(w,wv^i)\in \R$ for every $i\le n$ and there are the only worlds accessible from $w$ in $\M$.Thus, by induction hypothesis and semantics, we have $\M,w,\sigma \models \Box\beta_i$ for every $i\le m$ and $\M,w,\sigma \models \Diamond \phi_i$ for every $i\le n$. Also note that $(w,\Gamma,\sigma)$ is the last node of $w$ and hence, by definition of $\val$ at $w$ we have $\M,w,\sigma \models l_i$ for every $i\le s$. 

\item If ($\END$) rule is applied at $(w:\Gamma,\sigma)$ then $\Gamma$ is of the form  $\{l_1\ldots l_s\} \cup \{\Box\beta_1\ldots \Box\beta_m\}$ for some $s,m\ge 0$. There is one descendant $(w:\{l_1\ldots l_s\},\sigma)$ which is a leaf node in $T$.
Thus, there are no accessible worlds from $w$ in $\M$ and hence $\M,w,\sigma \models \Box\beta_j$ for all $j\le m$ vacuously. Finally since $(w:\{l_1\ldots l_s\},\sigma)$ is the last node of $w$, by definition of $\val$, we have $\M,w,\sigma \models l_j$ for every $j\le s$.
\end{itemize}

\bigskip

($\Leftarrow$) From Proposition \ref{prop-some rule can be applied always} it follows that we can always apply some rule until every leaf node $(w: \Gamma,\sigma)$ is such that $\Gamma$ contains only literals. Thus every (partial) tableau can be extended to a saturated tableau. To prove that such a tableau is open, it  suffices to show that all rules preserve satisfiability.  

The case of $(\land)$ and (\END) are trivial. For $(\lor)$ rule , if $\Gamma \cup\{\phi_1\lor\phi_2\}$  is satisfiable, then either $\Gamma\cup \{\phi_1\}$ is satisfiable or $\Gamma\cup \{\phi_2\}$ is satisfiable. Either way, $(\lor)$ rule also preserves satisfiability. For $(\Diamond)$ rule, it is straightforward by semantics. 

For $(\exists$)  rule, suppose $\M=(\W,\D,\live,\R,\val)$ is an increasing domain model, $w\in \W$ and $\pi$ is an assignment such that  $\M,w,\pi\models\exists x\phi\land \bigwedge \Gamma$. By semantics, it follows that there is a element $d\in \live(w)$ such that $\M,w,\pi_{[x\mapsto d]}\models\phi$. By cleanliness, $x$ is not free in  $\Gamma$ and hence $\M,w,\pi_{[x\mapsto d]}\models\bigwedge\Gamma$. Thus, $\{\phi\}\cup\Gamma$  is also satisfiable. 

For $(\forall)$ rule,   suppose $\M,w,\pi\models\forall y \phi\land \bigwedge\Gamma$.  By semantics, for every $d\in \live(w)$ we have  $\M,w,\pi_{[y\mapsto d]}\models\phi $. 

Let $Dom(\sigma) = \{z_0, z_1,\ldots z_k\}$ and for all $i\le k$ let $\pi(z_i) = d_i$. Now since every $\phi^*[z/y]$ is a clean rewriting of $\phi[z/y]$, we have $\M,w,\pi_{[z_0\mapsto d_0,\ldots z_k\mapsto d_k] }\models \phi^*[z_0/y]\land \ldots \phi^*[z_k/y]$. Also, by cleanliness for every $i\le k$ the variable $z_i$ does not occur in $\bigwedge\Gamma$. Hence we also have $\M,w,\pi_{[z_0\mapsto d_0,\ldots z_k\mapsto d_k]} \models \bigwedge \Gamma$.

Hence the set of formulas $\{\phi^*[z/y]\mid z\in Dom(\sigma)\}\cup \Gamma$ is satisfiable.  

\end{proof}
 
Note that the depth of the tableau is linear in the size of the formula. However, as we have to rewrite formulas using new variables when applying $(\forall)$ rule, the size of the domain  is exponential in the size of the formula. Therefore, the tableau procedure can be implemented by an \EXPspace algorithm. 

  From Proposition \ref{prop-LBF embeds many bundled operators},  $\ABEB$ and $\BABE$ are subfragments of $\LBF$. As a corollary, these fragments are also decidable. Also note that the quantifier prefix in $\LBF$ is of the form $\exists^*\forall^*$ and hence any  extension of this quantifier prefix or extending $\LBF$ with negation closure will result in a fragment that will be able to express $\forall x\exists y\Box~ \alpha$ (and hence will not have the finite model property,  cf Theorem \ref{thm-AEB has no FMP}).
In this sense  $\LBF$ is  the largest fragment in which $\forall x\exists y\Box~\alpha$ is not (syntactically) expressible.

\subsection{$\ABBABE$ over increasing domain} 
\label{sec: ABABE}

In this fragment we are allowed $\forall x\Box~\alpha,~ \Box\forall x~\alpha$ and $\Box \exists x~\alpha$ and their duals. 
Note that $\ABBABE$ fragment is not closed under subformulas. For instance, $ \phi ~:=~\forall x~\Big( \exists y\Diamond \alpha \lor \forall z \Box \beta\Big)$ is a subformula of $\phi'~:=~\Diamond \forall x~\Big( \exists y\Diamond \alpha \lor \forall z \Box \beta\Big)$. But $\phi'$ is in the fragment and $\phi$ is not in the fragment. 

We say that $\phi$ is a subformula of $\ABBABE$ if there is some formula $\phi' \in \ABBABE$ such that $\phi$ is a subformula of $\phi'$.

\begin{proposition}
\label{prop-no AEB inside ABBABE}
Let $\phi$ be a subformula of $\ABBABE$ such that $\phi$ is of the form $Q x~\psi$ where $Q\in \{\forall,\exists\}$.
 Then for every $\beta \in \comp(\psi)$,  $\beta$ is a \module or $\beta$ is of the form $\forall x \Box\beta'$ or $\exists x\Diamond \beta'$.
\end{proposition}

From Proposition \ref{prop-no AEB inside ABBABE}, it follows that we cannot (syntactically) express formulas of the form $\forall x\exists y\Box \phi$ in the fragment. However, $\forall x\exists y \Diamond \phi$ is still allowed. For instance, $\Diamond \forall x~\Big( \exists y\Diamond \alpha \Big)$ formula is in the fragment.

Thus, we can have $\forall x\exists y \Diamond \phi$ but not  $\forall x\exists y \Box\phi$.
Intuitively this means that the different witnesses $y$ for each $x$ can work on \textit{different} successor world. The fragment cannot enforce the interaction between $x$ and $y$ at all successors. This property can be used to prove that we can reuse the witnesses by creating new successor subtrees as required. 

 To get the decidability for $\ABBABE$ fragment, the main idea is to prove that the formulas of the form $\forall x\exists y\Diamond \phi$ can be satisfied by picking some boundedly many witnesses $y$ that will work for \textit{all} $x$. This is the same as proving that if $\forall x\exists y\Diamond \phi$ is satisfiable then $\exists y_1,\ldots \exists y_l \forall x \big( \bigvee \Diamond \phi[y/y_i] \big)$ is satisfiable (where $l$ is bounded).
We illustrate the proof idea with an example.

\begin{example}
 Consider the formula $ \alpha :=~\forall x\Big( \Box\neg Pxx \land \exists y \Diamond Pxy \Big)$ which is a subformula of the  fragment. Let $\T,r \models \alpha$ where $\T$ is a tree model rooted at $r$. Now we will modify $\T$ to obtain $\M$  which is also a tree model rooted at $r$ such that $\M,r\models \exists y_1 \exists y_2 \forall x~\Big( \Box \neg Pxx \land \big(~ \Diamond Pxy_1 \lor \Diamond Pxy_2~\big)\Big)$.

The model $\M$ is obtained by extending $\T$ in the following way. Let $\live^\T(r) = \D_r$. To obtain $\M$, first we extend the local domain of $r$ by adding a fresh element $a$. The idea is that for every $d\in \D_r$ (when assigned to $x$) we will ensure that the new element $a$ can be picked as the $y$-witness. To achieve this, we do the following: For every $d\in \D_r$ let $d' \in \D_r$ and $(r,s^d)\in \R^\T$ such that $\T,s^d \models Pdd'$. Let $\T^d$ be the subtree of $\T$ rooted at $s^d$. We will create a new copy of $\T^d$ and call its root as $u^d$. Now, in the  new subtree rooted at $u^d$, we make the new element $a$ `behave' like $d'$ and we add an edge from $r$ to $u^d$. So, in particular, $\M$  will have $(r,u^d)\in \R^\M$ such that $\M,u^d \models Pda$. Since we do this construction for every $d\in \D_r$ we obtain that for all $d\in \D_r$ we have $\M,r\models \Diamond Pda$.

Now note that  while evaluating $\alpha$ at $(\M,r)$ the $\forall x$ quantification will now also apply to $a$  (since $a$ is added to the local domain at $r$ in $\M$). But then, we cannot use $a$ itself as the witness for $a$  since we also need to ensure that $\M,r \models \forall x \Box \neg Pxx$. Hence we will add another element $b$ that acts as a witness for $a$. Further, $b$ also needs a witness. But now we can choose $a$ to be the witness for $b$ since that does not violate the formula $\forall x \Box \neg Pxx$.

So to complete the construction, we pick some arbitrary $d\in \D_r$ for which we have some $d'\in \D_r$ and $(r,s^d) \in \R^\T$ such that $\T,s^d\models Pdd'$.
We create two copies of $\T^d$ (subtree rooted at $s^d$) and call their roots as $v^d$ and $w^d$ respectively. In the subtree rooted at $v^d$ we ensure that $a$ and $b$ `behave' like $d,d'$ respectively and in the subtree rooted at $w^d$ we ensure that $a$ and $b$ `behave' like $d',d$ respectively. In particular we have $\M,v^d \models Pab$ and $\M,w^d\models Pba$. Finally we add edges from $r$ to $v^d$ and from $r$ to $w^d$ in $\M$. 

Thus, we have: $\live^\M(r) = \live^\T(r) \cup \{a,b\}$ and $\M,r\models \exists y_1 \exists y_2 \forall x~\Big( \Box \neg Pxx \land \big(~ \Diamond Pxy_1 \lor \Diamond Pxy_2~\big)\Big)$. With the above construction, this assertion can be verified by assigning $y_1$ and $y_2$ to $a$ and $b$ respectively.
\end{example}

Note that in principle, it is possible for a $\exists$ quantified formula to occur in the scope of a $\forall$ quantifier as a boolean combination with other $\exists$ quantified formulas and modules. Moreover these additional formulas can assert some `type' information that may force us to pick additional witnesses. For example if the formula is 
$\forall x\Big[ \Big( \Box (\neg Pxx  \land Rx)\lor \Box (\neg Pxx  \land \neg Rx)\Big)~ \land \exists y \Diamond \big( Rx \implies (Pxy \land \neg Ry) ~\land~ \neg Rx \implies (Pxy \land  Ry) \big)\Big]$, then we need two initial $y$-witnesses $a_1,a_2$ where one is used for witness whose `type' is $\Box (\neg Pyy \land Ry)$ and other for witness whose `type' is $\Box (\neg Pyy \land \neg Ry)$ and we also need the corresponding additional witnesses $b_1,b_2$.
In general the formula can force us to pick witness of a particular `$1$-type' which means we might need exponentially many witnesses. 

Thus,  we need to replace one $\exists$ inside the scope of $\forall$ by $2 l$ many $\exists$ quantifiers outside the scope of $\forall$ where $l$ is bounded exponentially in the size of the given formulas.  We now prove this formally.

For any formula $\phi$ if $\alpha\in \comp(\phi)$ we denote this by $\phi[\alpha]$. This means that $\alpha$ does not occur inside the scope of any modality in $\phi$.  For any formula $\beta$ we denote $\phi[\beta/\alpha]$ obtained by rewriting $\phi$ where $\alpha$ is replaced by $\beta$. In particular we are interested in the case where $\alpha$ is of the form  $\exists y \Diamond\psi$. Thus we always consider $\phi[\exists y\Diamond \psi]$.

For every $l\ge 0$ if $\overline{y} = y_1,y_1'\ldots y_l,y_l'$ are fresh variables, we denote $\overline{y}\Diamond\psi$ to be the formula $\bigvee\limits_{i\le l}\big( \Diamond \psi[y_i/y] \lor \Diamond \psi[y_i'/y]\big)$  which is a big disjunction where each disjunct replaces $y$ in $\psi$ with one of $y_i$ or $y_i'$. Further, we denote $\phi[\overline{y}\Diamond\psi / \exists y \Diamond\psi]$ as simply $ \phi[\overline{y} \Diamond\psi]$.

For instance, for the formula $\phi := \Big(Px~ \lor~ \exists y\Diamond Qxy\Big)$ where $\psi := \exists y\Diamond Qxy$, for $l=2$ and $\overline{y} = y_1,y_1',y_2,y_2'$ being fresh variables, $\phi[\overline{y} \Diamond\psi]$ is given by:
 $\Big( Px~ \lor \big(\Diamond Qxy_1 \lor \Diamond Qxy_1' \lor \Diamond Qxy_2 \lor \Diamond Qxy_2'   \big)\Big)$

The size of a formula denoted by $|\phi|$ is the number of symbols occurring in $\phi$ and for a finite set of formulas $\Gamma$, let  $|\Gamma| = \sum\limits_{\phi\in \Gamma}|\phi|$.

\begin{lemma}
\label{lemma-bounded witness enough}
Let $\Gamma'$ be  a clean finite set of  formulas such that every $\alpha\in \Gamma$ is a subformula of $\ABBABE$ where $\Gamma' = \Gamma \cup \{\forall x \phi[\exists y \Diamond \psi]\}$. 
 If~ $\bigwedge\Gamma \land \forall x \phi[\exists y \Diamond \psi]$ is satisfiable then \\there exists $l \le 2^{|\Gamma'|}$ such that $\bigwedge\Gamma~ \land~ \exists y_1\exists y_1'\exists y_2\exists y_2'\ldots \exists y_l\exists y_l'~ \forall x~ \phi[\overline{y} \Diamond\psi]$ is satisfiable, where  $\overline{y} = y_1,y_1',\ldots y_l,y_l'$ are fresh variables.
\end{lemma}

Note that the $\exists y$ quantifier is pulled outside the scope of the $\forall x$ quantifier and replaced with bounded number of witnesses $y_1,y_1',y'_2,y_2'\ldots y_l,y_l'$. Consequently $\Diamond~\psi$ is replaced with a disjunction each replacing $y$ with one of $y_i$ or $y_i'$ for every $i\le l$.

To prove the lemma first we formally define the tree editing operation described in the example. Given a tree model $\T$ rooted at $r$, let $d\not\in \D^\T$. To add the new domain element $d$ to a local domain of $r$, we also need to specify the `type' of the new element $d$ at $r$ and its descendants. Towards this, we pick some domain element $c$ that is already present in $\live(r)$ and assign the type of $d$ to the type of $c$ at every world.

\begin{definition}
\label{def-tree operation}
Given a tree model $\T = (\W,\D,\R,\live,\val)$ rooted at $r$, let $d\not\in \D$ and $c\in \live(r)$. Define the operation of `adding $d$ to $\live(r)$ by mimicking $c$', denoted by $\T_{d\mapsto c} = (\W,\D',\R,\live',\val')$ where:
\begin{itemize}
\item[-] $\D' = \D\cup \{d\}$ 
\item[-] for all $w\in \W$ we have $\live'(w) = \live(w) \cup \{d\}$
\item[-] For every $w\in \W$ and predicate $P$ define\\
 $\val'(w,P) =  \{ \overline{e'} \mid$ there is some $\overline{e} \in \live(w,P)$ and $\overline{e'}$ is obtained from $\overline{e}$ by replacing zero or more occurrences of $c$ in $\overline{e}$ by $d\}$
 \end{itemize}

Suppose that we want to extend the domain  with $\overline{d} = d_1,\ldots d_n$ which are fresh. Let $\omega: \overline{d}\mapsto \D'$ where $\D'\subseteq \delta^\T(r)$ and we want each $d_i$ to mimic $\omega(d_i)$. Then we  denote $\T_{\omega}$ to be the tree obtained by $\Big( \big(\T_{d_1\mapsto \omega(d_1)}\big)_{d_2\mapsto \omega(d_2) \ldots} \Big)_{d_n\mapsto \omega(d_n)}$.
 
\end{definition}

\begin{proposition}
\label{prop-tree extension is ok}
Let $\T = (\W,\D,\R,\live,\val)$ be a tree model rooted at $r$ with $\T_{d\mapsto c}$ being an extended tree where $d\not\in \D$ and $c\in \live(r)$. Then for all interpretation $\sigma$ and for all $\FOML$ formula $\phi$ and all  $w \in \W$:\\
{\small $\T,w,\sigma_{[x\mapsto c]} \models \phi$ ~~iff~~ $\T_{d\mapsto c},w,\sigma_{[x\mapsto c]} \models \phi$ ~~iff~~ $\T_{d\mapsto c},w,\sigma_{[x\mapsto d]} \models \phi$}
\end{proposition}
\begin{proof}
Proof is by induction on the structure of $\phi$. The atomic case follows by the definition of $\live'$. The $\neg$ and $\land$ cases are standard. 

\begin{itemize}
\item 
For the $\Diamond \phi$ case, if $\T,w,\sigma_{[x\mapsto c]} \models \Diamond \phi$ then there is some $w\to v$ such that $\T,v,\sigma_{[x\mapsto c]} \models \phi$. By induction,  $\T_{d\mapsto c},v,\sigma_{[x\mapsto c]}\models \phi$ and  $T_{d\mapsto c},v,\sigma_{[x\mapsto d]} \models \phi$. Hence $T_{d\mapsto c},w,\sigma_{[x\mapsto c]}\models \Diamond \phi$ and $T_{d\mapsto c},w,\sigma_{[x\mapsto d]} \models \Diamond \phi$.

Conversely, if $\T_{d\mapsto c},w,\sigma_{[x\mapsto d]}\models \Diamond \phi$ then there is some $w\to v$ such that $\T_{d\mapsto c},v,\sigma_{[x\mapsto d]}\models \phi$. By induction $\T,v,\sigma_{[x\mapsto c]} \models \phi$ and $\T_{d\mapsto c},v,\sigma_{[x\mapsto c]} \models \phi$. Hence $\T,w,\sigma_{x\to c}\models \Diamond \phi$ and $\T_{d\mapsto c},w,\sigma_{[x\mapsto c]}\models \Diamond \phi$.

\item  For the case $\exists y~ \phi$, if $\T, w,\sigma_{[x\mapsto c]}\models \exists y~ \phi$ then let $c'\in \delta^T(r)$ such that $\T,w,\sigma_{[xy\to cc']}\models \phi$. By induction $\T_{d\mapsto c},w,\sigma_{[xy\to cc']}\models \phi$ and $\T_{d\mapsto c},w,\sigma_{[xy\to dc']}\models \phi$.
 Thus, $\T_{d\mapsto c},w,\sigma_{[x\mapsto c]}\models \exists y~ \phi$ and $\T_{d\mapsto c},w,\sigma_{[x\mapsto d]}\models \exists y~ \phi$.

\medskip
Conversely, suppose $\T_{d\mapsto c},w,\sigma_{[x\mapsto d]}\models \exists y~ \phi$ then there exists $d' \in \live^{\T_{d\mapsto c}}(r)$ such that $\T_{d\mapsto c},w,\sigma_{[xy\mapsto dd']}\models \phi$. Now, if $d'\ne d$ then $d'\in \live^T(r)$ and by induction we have $\T,w,\sigma_{[xy\mapsto cd']}\models \phi$ and $\T_{d\mapsto c},w, \sigma_{[xy\mapsto cd']}\models \phi$. Otherwise if $d'=c$ then also by induction we have $\T,w,\sigma_[{xy\mapsto cc]}\models \phi$ and $\T_{d\mapsto c},w,\sigma_{[xy\mapsto cc]}\models \phi$. So  we have both $\T,w,\sigma_{[x\to c]}\models \exists y~ \phi$ and $\T_{d\mapsto c},w,\sigma_{[x\mapsto c]}\models \exists y~ \phi$. 
\end{itemize}
\end{proof}

Now we are ready to prove Lemma \ref{lemma-bounded witness enough}.

\begin{proof}[Proof of Lemma \ref{lemma-bounded witness enough}]
 Let $\T$ be a tree model rooted at $r$ such that $\T,r,\sigma \models\bigwedge\Gamma \land \forall x \phi[\exists y \Diamond \psi]$.

For every domain element $a\in \delta^T(r)$ define:

\begin{center}
\small
\begin{tabular}{c r}
$\type(r,a) =$
 &$\bigcup\limits_{\forall x'\alpha \in \Gamma'} \{ \lambda \mid \lambda\in \comp(\alpha)$ and $T,r,\sigma_{[x' \mapsto a]} \models \lambda\}$\\ 
 \end{tabular}
 \end{center}

  Let $\type(r) = \{ \type(r,a) \mid a\in \delta^T(r)\}$.  Let $|\type(r)| = l $ (note that $l\le 2^{|\Gamma'|}$).
  Enumerate $\type(r) = \{ \Lambda_1,\ldots \Lambda_l\}$ and for every $i\le l$ pick $a_i \in \delta^T(r)$ such that  $\type(r,a_i)   = \Lambda_i$.

Now let $\overline{d} = d_1,d_1',d_2,d_2',\ldots d_l,d_l'$ be fresh domain elements and let $\omega: \overline{d} \mapsto \{a_1,a_2,\ldots a_l\}$ where for all $i\le l$ we have $\omega(d_i) = \omega(d_i') = a_i$.
We define the required model $\M = (\W',\D',\R',\live',\val')$ as follows:

Let $\M_0 = \T_{\omega}$ be the new tree model  rooted at $r$ obtained by adding $d_1,d_1',\ldots, d_l,d_l'$ to $\live^T(r)$  where each $d_i$ and $d_i'$ mimics $a_i$. Now $\M$ is obtained by extending $\M_0$ as follows:  

For every $c\in \delta(r)$  such that $\T,r,[x\mapsto c]\models \exists y\Diamond \psi$ we pick $c' \in \delta^T(r)$ and $r\to s^c$ be such that~~ $\T,s^c,[xy\mapsto cc'] \models \psi$. Let $\T^c$ be the sub-tree of $\T$ rooted at $s^c$  and $\type(r,c') = \Lambda_j$. Then :

 Create a new subtree $\Center^c = \T^c_{\omega'}$ where for all $h\ne j$ we have $\omega'(d_h) = \omega'(d_h') = a_h$ and $\omega'(d_j) = \omega'(d_j') = c'$. Let $u^c$ be the root of $\Center^c$. Add an edge from $r$ to $u^c$ in $\M$.
 
 Further, for every $i\le l$ if  $\T,r,[x\mapsto a_i]\models \exists y\Diamond \psi$, then let $b \in \delta^T(r)$ and $r\to s^{i} \in \R^T$ be  such that\\ $\T,s^{i},[xy\mapsto a_ib]\models  \psi$.  Let $\type(r,b) = \Lambda_j$. Then:
 
 Create  $\Left^i = \T^{i}_{\omega_1}$ and $\Right^i = \T^{i}_{\omega_2}$ where $\omega_1$ and $\omega_2$ are defined as follows:
 \begin{itemize}
 \item For all $h\ne j$,\\
 $\omega_1(d_h) = \omega_1(d_h')  = \omega_2(d_h) = \omega_2(d_h') = a_h$
 \item $\omega_1(d_j) = a_j$ and $\omega_1(d_j') = b$
 \item $\omega_2(d_j) = b$ and $\omega_2(d_j') = a_j$
 \end{itemize}
 Let $v^i$ and $w^i$ be the root of $\Left^i$ and $\Right^i$ respectively. Add the edges from $r$ to $v^i$ and from $r$ to $w^i$ in $\M$. The two copies of subtrees are intended to provide witness for $\exists y \Diamond \psi$ for $d_j$ and $d_j'$ respectively. We need the two copies to ensure that $\omega_1$ and $\omega_2$ are well defined in the case when $i=j$ and $a_j\ne  b$. 
 
 Note that $\Center^c$ rooted at $u^c$ is created for every $c\in \live^T(r)$ such that $\T,r,\sigma_{[x\mapsto c]} \models \exists y \Diamond \psi$. If $c'$ is the picked witness for $c$ with $(r,s^c)\in \R^T$  such that $\T,s^c,\sigma[xy\mapsto cc'] \models \psi$ and $\type(r,c') = \Lambda_j$ then by construction $\Center^c$ rooted at $u^c$ is obtained where $d_j$ mimics $c'$ in the subtree rooted at $u^c$. All these together indicate that we can use $d_j$ and the subtree rooted at $u^c$ in $\M$  to verify that $\M,u^c,\sigma_{[xy\mapsto cd_j]} \models \psi$. Also note that for all $h\ne j$ the fresh elements $d_h$ and $d_h'$ mimic $a_h$ in the subtree $\Center^c$ (i.e, we have not added any extra `types').
 
Further, we want the type of $d_i$ and $d_i'$ at $r$ in $\M$ to mimic the type of $a_i$ at $r$ in $\T$. All  type information is taken care in $\M_0$ where both $d_i$ and $d_i'$ mimic $a_i$ except the formula $\exists y \Diamond \psi$. So if $\T,r,\sigma_{[x\mapsto a_i]}\models  \exists y \Diamond \psi$ then we need witness to verify $\M,r,\sigma_{[x\mapsto d_i]}\models  \exists y \Diamond \psi$  and $\M,r,\sigma_{[x\mapsto d_i']}\models \exists y \Diamond \psi$.  If the witness for $y$ for $a_i$ is $b$ and $\type(r,b) = \Lambda_j$ then we want $d_j'$ to be the witness for $d_i$ and $d_j$ to be the witness for $d_i'$. 

Consequently if $s^i$ is the world such that $a\to s^i \in \R^T$ and $\T,s^i,[xy\mapsto a_ib] \models \psi$ then we create two new copies of subtree $\T^i$ rooted at $s^i$ and call it $\Left^i$ and $\Right^i$. By construction, in particular, the new element $d_i$ mimics $a_i$  and $d_j'$ mimics $b$ in $\Left^i$. Similarly $d_i'$ mimics $a_i$ and $d_j$ mimics $b$ in $\Right^i$. Thus, we can pick $d_j'$ to be the witness for $d_i$ (and consider $\Left^i$) and pick $d_j$ to be the witness for $d_i'$ (and consider $\Right^i$). 

Also, it is important to note that  for every $r\to v\in \R^M$, if $d_i$ mimics $c$ and $d_i'$ mimics $c'$ at $v$ then we will always have $\type(r,c) = \type(r,c') = \Lambda_i = \type(r,a_i)$.
Now it can be verified that
$\M,r,\sigma \models \bigwedge\Gamma~\land~\exists y_1 \exists y_1' \ldots \exists y_l\exists y_l' ~ \forall x \phi[\overline{y} \Diamond \psi]$.

Towards this, first we prove some useful claims.

\bigskip

{\it Claim 1.} \label{Claim: old interpretation is ok for modules}
 For every interpretarion $\pi: Var\mapsto \delta^T(r)$ and every \module $\alpha$,~~
 If $\T,r,\pi \models \alpha$ then $\M,r,\pi\models \alpha$.
 
 \medskip
 {\it Proof.} Let $\alpha$ be a module. We consider all ppssible cases.
 If $\alpha$ is a literal then the claim follows since $\val^M(r)$ is an extension of $\val^T(r)$.

For the case $\Box ~\alpha'$ let $\T,r,\pi \models \Box\alpha'$. To verify  $\M,r,\pi \models \Box \alpha'$ pick any arbitrary successor $r\to v \in \R^M$. Let $v$ be a copy of $w \in \W^T$ extended with $d_1,d_1',\ldots d_l,d_l'$. Since $\T,r,\pi\models \Box\alpha'$, we have $\T,w,\pi \models \alpha'$. By Proposition \ref{prop-tree extension is ok} we have $\M,v,\pi \models \alpha'$. 

Since we picked $v$ arbitrarily, we have $\M,r,\pi\models \Box\alpha'$.

\medskip

For the case $\Diamond ~\alpha'$ let $\T,r,\pi \models \Diamond\alpha'$. By semantics, there exists $r\to w \in \R^T$ such that $\T,w,\pi \models \alpha'$. Then by construction the subtree of $\T$ rooted at $w$ continues to be a part of $\M_0$ (and hence $\M$) extended with $d_1,d_1',\ldots d_l,d_l'$. So by Proposition \ref{prop-tree extension is ok} we have $\M,w,\pi\models \alpha'$ and $(r,w)\in R^M$. Thus, $\M,r,\pi \models \Diamond \alpha'$.

\bigskip
{\it Claim 2.}
\label{Claim: old interpretation is ok for all formulas inside forall scope}
 Let $\pi: Var\to \delta^T(r)$ be any interpretation. Let $\alpha$ be a \module or of the form $\forall z\Box~\beta$ or $\exists z\Diamond~\beta$. ~~~
  If $\T,r,\pi \models \alpha$ then $\M,r,\pi \models \alpha$.

\medskip
{\it Proof.} 
 We consider all possible cases of $\alpha$:

\begin{itemize}
\item If $\alpha$ is a module then Claim $2$ follows from Claim $1$.

\item If $\alpha = \forall z\Box~ \beta$ then by semantics, for every $d\in \live^T(r)$ and every $r\to w\in R^T$ we have $\T,w,\pi_{[z\mapsto d]}\models \beta$.

Suppose $\M,r,\pi\not\models \forall z\Box~\beta$ then by semantics we have  $\M,r,\pi\models \exists z\Diamond(\neg \beta)$. Let $c\in \live^M(r)$ and $r\to v \in \R^M$ such that $\M,v,\pi_{[z\mapsto c]}\models \neg\beta$. 

Let the subtree rooted at $v$ in $\M$ be a copy of the subtree rooted at $w\in \W^T$ in $\T$. Now if $c\in \live^T(r)$ then by proposition \ref{prop-tree extension is ok} we have $\T,w,\pi_{[z\mapsto c]}\models \neg \beta'$ which is a contradiction. 

Otherwise if $c$ mimics some $c'$ at the world $v$ then again by Proposition \ref{prop-tree extension is ok} we have $\T,w,\pi_{[z\mapsto c']}\models \neg \beta'$ which is a contradiction.

\item If $\alpha = \exists z\Diamond~\beta$ then let $d\in \delta^T(r)$ and $r\to w\in \R^T$ such that $\T,w,\pi_{[z\mapsto d]} \models \beta$.  By construction, the subtree rooted at $w$ in  $\T$  continues to be a part of $\M_0$ (and hence $\M$) extended with $d_1,d_1',\ldots d_l,d_l'$. 

So by Proposition \ref{prop-tree extension is ok}  we have $\M,w,\pi_{[z\mapsto d]}\models \beta$ and $(r,w)\in \R^M$. Thus, $\M,r,\pi_{[z\mapsto c]} \models \exists z\Diamond \beta$ which is a contradiction.
\end{itemize}

 \bigskip
 {\it Claim 3.} 
\label{Claim: new elements are ok for all subformulas}
 Let $d\in \{d_k,d_k'\}$ for some $k\le l$. Then for every $\forall z~\alpha~~\in \Gamma'$ and for every $\beta\in \comp(\alpha)$:
  If $\T,r,\sigma_{[z \mapsto a_k]} \models \beta'$ then $\M,r,\sigma_{[z\mapsto d]}\models \beta'$.

\medskip
{\it Proof. }  From Proposition \ref{prop-no AEB inside ABBABE} it follows that every $\beta\in \comp(\alpha)$ is either a \module or of the form $\forall z'\Box\beta'$ or $\exists z'\Diamond \beta'$. We consider various possible cases for $\beta$:

\begin{itemize}
\item If $\beta$ is a \literal then the claim follows since $d$ mimics $a_k$ at $r$ and $\val^M$ is defined accordingly.
\item If $\beta$ is of the form $\Box\beta'$ then by definition, $\Box\beta'\in \type(r,a_k)$ and hence for all $c\in \delta^T(r)$ if $\type(r,c) = \type(r,a_k)$ then $\T,r,\sigma[z\mapsto c]\models \Box\beta'$. Thus, for every $c$ such that $\type(r,a_k) = \type(r,c)$ and for every $r\to w \in \R^T$ we have $\T,w,\sigma[z'\mapsto c] \models \beta'$. 

Now to verify that $\M,r,\sigma[z\mapsto d]\models \Box \beta'$, pick some arbitrary world $v$ such that $r\to v\in \R^M$. Let the subtree rooted at $v$ in $\M$ be a copy of the subtree rooted at $w\in \W^T$ in $\T$.  Let $d$ mimic some $c$ in the subtree rooted at $v$ in $\M$. By construction it follows that $\type(r,c) = \type(r,a_k)$. So we have $\T,w,\sigma_{[x\mapsto c]}\models \beta''$ and by Proposition  \ref{prop-tree extension is ok} we have $\T,w,\sigma_{[z\mapsto d]}\models  \beta'$.

\item If $\beta$ is of the form $\Diamond \beta'$ then there is some $r\to w \in \R^T$ such that $\T,w,\sigma_{[z\mapsto a_k]} \models \beta'$. Then by construction, the subtree rooted at $w$ is a subtree of $\M_0$ (and hence subtree of $\M$) where $d$ mimics $a_k$. Hence by proposition \ref{prop-tree extension is ok} we have $\M,w,\sigma_{[z\mapsto d]} \models \beta'$ which implies $\M,r,\sigma_{[z\mapsto d]} \models \Diamond \beta'$.

\item If $\beta$ is of the form $\exists z' \Diamond \beta'$ then let $c'\in \delta^T(r)$ and $r\to w\in \R^T$ such that $\T,w,\sigma_{[zz'\mapsto a_kc']} \models \beta'$. By construction, the subtree of $\T$ rooted at $w$ continues to be a part of $\M_0$ (and hence of $\M$) where $d$ mimics $a_k$. Hence, by Proposition \ref{prop-tree extension is ok} we have $\M,w,\sigma_{[zz'\mapsto dc']}\models \beta'$. Thus, $\M,r,\sigma_{[z\to d]}\models \exists z'\Diamond \beta'$.

\item If $\beta$ is of the form $\forall z' \Box \beta'$ then then pick some arbitrary $c'\in \delta^M(r)$ and some arbitrary $r\to v\in \R^M$.  We will verify that $\M,v,\sigma_{[zz'\mapsto dc']}\models \beta'$.  Let $v$ be a copy of $w\in \W^T$ and let $d$ mimic some $e$ at $v$. Then by construction we have $\type(r,a_k) = \type(r,e)$ and hence $\T,r,\sigma_{[z\mapsto e]}\models \forall z' \Box\beta'$.

Now if $c' \in \delta^T(r)$ then we have $\T,w,\sigma_{zz'\mapsto ec'}\models  \beta'$ and since $d$ mimics $e$ at $v$, by Proposition \ref{prop-tree extension is ok} we have $\M,v,\sigma_{[zz'\mapsto dc']}\models \beta'$.

Otherwise if $c'\in \{d_1,d_1',\ldots d_l,d_l'\}$ then let $c'$ mimic some $d'$ at $v$ and again we have $\T,w,\sigma_{[zz'\mapsto ed']}\models \beta'$. Further, since $d,c'$ mimics $e,d'$ respectively at $v$, by Proposition \ref{prop-tree extension is ok} we have $\M,v,\sigma_{[zz'\mapsto dc']}\models \beta'$.
\medskip

Since we picked $c'$ and $v$ arbitrarily, we have $\M,r,\sigma_{[z\mapsto d]}\models \forall z'\Box\beta'$.
\end{itemize}

Now we are ready to prove that
 $\M,r,\sigma \models \bigwedge\Gamma \land \exists y_1\exists y_2\ldots \exists y_l\exists y_l' \forall x ~\phi[\overline{y} \Diamond \psi]$
\\
First pick some $\alpha \in \Gamma$. We have $\T,r,\sigma\models \alpha$ and we need to verify that $\M,r,\sigma \models \alpha$. Suppose not then there is some $\chi\in \comp(\alpha)$ such that $\T,r,\sigma \models \chi$ and $\M,r,\sigma \models \neg \chi$ where $\chi$ is a \module or of the form $Qz~\alpha'$.

 Note that $\sigma$ is a mapping from $\Var$ to  $\delta^T(r)$. So, if $\chi$ is a \module then the assumption is a contradiction to Claim $1$.
Otherwise $\chi$ is of the form $Qz~\alpha'$. Then by Proposition \ref{prop-no AEB inside ABBABE} every $\beta\in \comp(\alpha')$ is either a \module or of the form $\forall z'\Box\beta'$ or $\exists z' \Diamond \beta'$. Now we have two cases:

\begin{itemize}
\item If $\chi$ is of the form $\exists z~\alpha'$ then let $c\in \live^T(r)$ such that $\T,r,\sigma_{[z\mapsto c]}\models \alpha'$. But now if $\M,r,\sigma_{[z\mapsto c]}\not\models \alpha'$ then there is some $\beta'\in \comp(\alpha')$ such that $\T,r,\sigma_{[z\mapsto c]}\models \beta'$ and $\M,r,\sigma_{[z\mapsto c]}\models \neg \beta'$. This is a contradiction to Claim $2$.

\item If $\chi$ is of the form $\forall z~\alpha'$ then  pick any arbitrary $d\in \live^M(r)$ and we verify that $\M,r,\sigma_{[z\mapsto c]}\models \alpha'$ (this implies $\M,r,\sigma \models \forall z~\alpha'$ which contradicts the assumption)

Suppose $d\in \live^T(r)$ then by Claim $2$, for every $\beta' \in \comp(\alpha')$
if $\T,r,\sigma_{[z\mapsto d]}\models \beta'$ then $\M,r,\sigma_{[z\mapsto d]}\models \beta'$. Hence, $\M,r,\sigma_{[z\mapsto d]}\models \alpha'$.

\medskip

Suppose $d\in \{d_1,d_1',\ldots d_l,d_l'\}$ then let $d \in \{d_k,d_k'\}$ for some $k\le l$. Note that we have $\T,r,\sigma_{[z\mapsto a_k]}\models \alpha'$. By Claim $3$, for every $\beta' \in \comp(\alpha')$
if $\T,r,\sigma_{[z'\mapsto a_k]}\models \beta'$ then $\M,r,\sigma_{[z'\mapsto d]}\models \beta'$. Hence, $\M,r,\sigma_{[z\mapsto d]}\models \alpha'$.

\end{itemize}

Thus we have proved that $\M,r,\sigma \models \bigwedge \Gamma$.

\bigskip

Now we verify  
$\M,r,\sigma \models \exists y_1 \exists y_1' \ldots \exists y_l\exists y_l' ~ \forall x \phi[\overline{y} \Diamond \psi]$. Let $d_1,d_1',\ldots d_l,d_l'$ be the witness for $y_1,y_1',\ldots y_l,y_l'$ respectively.
Let $\sigma'$ be an extension of $\sigma$ where for all $i\le l$ we have $\sigma'(y_i) = d_i$ and $\sigma'(y_i') = d_i'$. Thus, we will verify that $M,r,\sigma' \models \forall x~ \phi[\overline{y} \Diamond \psi]$.
Pick any arbitrary $c\in \live^M(r)$. We claim that $\M,r,\sigma'_{[x\mapsto c]} \models \phi[\overline{y} \Diamond \psi]$.  Again we have two cases:

\begin{itemize}
\item Suppose $c\in \live^T(r)$ then note  we have $\T,r,\sigma_{[x\mapsto c]}\models \phi$. So, if the claim is false then there is some $\chi \in \comp(\phi)$ such that $\T,r,\sigma_{[x\mapsto c]}\models \chi$ and $M,r,\sigma_{[x\mapsto c]}\not\models \hat{\chi}$ where $\hat{\chi} = \chi$ if $\chi\ne \exists y\Diamond~\psi$ and otherwise $\hat{\chi} = \overline{y}\Diamond\psi$.

Also by Proposition \ref{prop-no AEB inside ABBABE}, $\chi$ is either a \module or of the form $\forall z\Box \beta$ or $\exists z\Diamond \beta$. 
Now, if $\chi\ne  \exists y\Diamond~\psi$ then $\hat{\chi} = \chi$ and by Claim $2$, since $\T,r,\sigma_{[y\mapsto c]}\models \chi$ we have $\M,r,\sigma_{[y\mapsto c]}\models \chi$ which is a contradiction to the assumption.

So let $\chi = \exists y\Diamond\psi$. Then let $c'\in \live^T(r)$ and $r\to s^c \in \R^T$ be such that $\T,s^c,\sigma_{[xy\mapsto cc']}\models \psi$. Let $\type(r,c') = \type(r,a_k)$ for some $k\le l$. By construction we have a subtree $\Center^c$ of $\M$ rooted at $u^c$ which is a copy of $\T^c$ (the subtree rooted at $s^c$ in $\T$) extended with $d_1,d_1',\ldots d_l,d_l'$ where $d_k$ mimics $c'$. So by Proposition \ref{prop-tree extension is ok} it follows that $\M,u^c,\sigma_{[xy\mapsto cd_k]}\models \psi$ which implies $\M,u^c,\sigma'_{[x\to c]} \models \psi[y_k/y]$. Hence $\M,r,\sigma'_{[x\mapsto c]} \models \overline{y}\Diamond\psi$.

\bigskip

\item Otherwise $c\in \{d_1,d_1',\ldots d_l,d_l'\}$. Let $c\in \{d_j,d_j'\}$ for some $j\le l$. Note that we have $\T,r,\sigma_{[x\mapsto a_j]}\models \phi$. So, if the claim is false then there is some $\chi \in \comp(\phi)$ such that $\T,r,\sigma_{[x\mapsto d_j]}\models \chi$ and $\M,r,\sigma'_{[x\mapsto c]}\not\models \hat{\chi}$ where $\hat{\chi} = \chi$ if $\chi\ne \exists y\Diamond~\psi$ and otherwise $\hat{\chi} = \overline{y}\Diamond\psi$.

Also by Proposition \ref{prop-no AEB inside ABBABE}, $\chi$ is either a \module or of the form $\forall z\Box \beta$ or $\exists z\Diamond \beta$. 
Now, if $\chi\ne  \exists y\Diamond~\psi$ then  $\hat{\chi} = \chi$ and by Claim $3$, since $\T,r,\sigma_{[x\mapsto a_j]}\models \chi$ we have $\M,r,\sigma'_{[x\mapsto c]}\models \chi$ which is a contradiction to the assumption (note that $\sigma$ and $\sigma'$ differ only on the valuation of $y_1,y_1',\ldots y_l,y_l'$ and these variables do not occur in $\chi$ in this case).

So let $\chi = \exists x\Diamond\psi$. Then let $b\in \live^T(r)$ and $r\to s^j \in \R^T$ be such that $\T,s^j,\sigma_{[xy\mapsto a_jb]}\models \psi$. Let $\type(r,b) = \type(r,a_k)$ for some $k\le l$. 

\medskip
Now by construction, we have the subtree $\Right^j$ rooted at $v^j$ in $\M$ which is a copy of $\T^j$ (subtree rooted at $s^j$ in $\T$) such that $d_j$ and $d_k'$ mimic $a_k$ and $b$ respectively in $\Right^j$. Hence by Proposition \ref{prop-tree extension is ok} we have $\M,v^j,\sigma_{[xy\mapsto d_jd_k']}\models \psi$ which implies $\M,v^j,\sigma'_{[x\mapsto d_j]}\models \psi[y_k'/y]$. Hence, we have $\M,r,\sigma'_{[x\mapsto d_j]}\models \overline{y}\Diamond\psi$.

Similarly, we have the subtree $\Left^j$ rooted at $w^j$ and $w^j$ in $\M$ which is a copy of $\T^j$ (subtree rooted at $s^j$ in $\T$) such that $d_j'$ and $d_k$ mimic $a_k$ and $b$ respectively in $\Left^j$. Hence by Proposition \ref{prop-tree extension is ok} we have $\M,v^j,\sigma_{[xy\mapsto d_j'd_k]}\models \psi$ which implies $\M,v^j,\sigma'_{[x\mapsto d_j']}\models \psi[y_k/y]$. Hence, we have $\M,r,\sigma'_{[x\mapsto d_j']}\models \overline{y}\Diamond\psi$.

\end{itemize}

Thus, in both cases of $c\in \{d_j,d_j'\}$ we have  $M,r,\sigma'_{[x\mapsto c]} \models \overline{y}\Diamond\psi$ (which is a contradiction to the assumption).
 \end{proof}

\begin{corollary}
Let $\Gamma'$ be  a clean finite set of  formulas such that every $\alpha\in \Gamma$ is a subformula of $\ABBABE$ where $\Gamma' = \Gamma \cup \{\forall x \phi[\exists y \Diamond \psi]\}$. 
 If~ $\bigwedge\Gamma \land \forall x \phi[\exists y \Diamond \psi]$ is satisfiable then  $\bigwedge\Gamma~ \land~ \exists y_1\exists y_1'\exists y_2\exists y_2'\ldots \exists y_l\exists y_l'~ \forall x~ \phi[\overline{y} \Diamond\psi]$ is satisfiable, where  $l = 2^{|\Gamma'|}$ and  $\overline{y} = y_1,y_1',\ldots y_l,y_l'$ are fresh variables.
\end{corollary} 

 To see why the corollary is true, by Lemma \ref{lemma-bounded witness enough} we get some $l \le 2^{|\Gamma'|}$, and we can pad sufficiently many dummy variables to get a strict equality. This  gives us a useful tableau rule which we call ($\forall\exists\Diamond$) rule for $\ABBABE$ fragment, described in figure \ref{fig-tableau for ABBABE}. The full tableau rules for $\ABBABE$ is given by the tableau rules of $\LBF$ (Figure \ref{fig-tableau for LBF}) along with the ($\forall\exists \Diamond$)-rule.

\begin{figure}[h]
 \begin{center}
 \begin{tabular}{|c c|}
\hline
        \multicolumn{2}{|c|}{}\\
      \multicolumn{2}{|c|}{$\dfrac{w: \forall x~\phi[\exists y\Diamond \psi],~ \Gamma,~\sigma}
       {w:\forall x~\phi[\overline{y} \Diamond\psi],~ \Gamma,~\sigma'}(\forall\exists\Diamond)$} \\& \\
        \multicolumn{2}{|c|}{ where  $l= 2^{|\Gamma|+|\phi|}$ and $\overline{y} = y_1,y_1',\ldots y_l,y_l'$ }\\
        \multicolumn{2}{|c|}{ are fresh variables and $\sigma' = \sigma \cup \{ (y_i,y_i)~,(y_i',y_i') \mid i\le l\}$}\\
      
           \hline
\end{tabular}

\caption{($\forall \exists\Diamond$) rule for $\ABBABE$ fragment}
\label{fig-tableau for ABBABE}
\end{center}
\end{figure}

\begin{theorem}
\label{thm-tableau for ABBABE}
For any clean $\ABBABE$ formula $\theta$, let $Dom(\sigma_r)=Free(\theta)\cup\{z\}$ where $z$ does not occur in $\theta$. There is an open tableau with $(r:\{\theta\},\sigma_r)$ as the root iff $\theta$ is satisfiable in an increasing domain model. 
\end{theorem}

\begin{proof}
First note that $(\forall\exists\Diamond)$ preserves cleanliness because the new free variables $y_1,y_1',\ldots y_l,y_l'$ are fresh  and hence not bound by any other quantifier in the antecedent. 

($\Rightarrow$) Let $\T$ be an open tableau rooted at $(r:\{\theta\},\sigma_r)$. The model is defined exactly as we defined in the Proof of Theorem \ref{thm-tableau for LBF} and we have the claim:

 For every node $(w:\Gamma,\sigma)$ in $T$  and for every $\LBF$ formula $\phi$ if $\phi\in \Gamma$ then  $\M,w,\sigma\models \phi$.
 
The proof of the claim is by induction on the nodes of $T$ from leaf nodes to the root which same as in the proof of Theorem \ref{thm-tableau for LBF}.
The only additional part is to prove the inductive claim  for the application of  $(\forall\exists\Diamond)$ rule at a node $(w: \Gamma,\sigma)$.

In this case $\Gamma$ is of the form $\Gamma'\cup \{\forall x \phi[\exists y \Diamond \psi]\}$ and we have a descendant $(w: \Gamma' \cup \{\forall x\phi[\overline{y}\Diamond \psi]\}, \sigma')$ where $\sigma' = \sigma \cup\{(y_i,y_i),~(y_i',y_i')\mid i\le l\}$.

First note that since $y_1,y_1'\ldots y_l,y_l'$ are fresh,  by induction hypothesis, for every $\phi'\in \Gamma'$ we have $\M,w,\sigma \models \phi'$.
Now to prove that $\M,w,\sigma \models \forall x \phi[\exists y \Diamond \psi]$, pick any arbitrary $z\in Dom(\sigma)$. We claim that $\M,w,\sigma_{[x\mapsto z]}\models \phi[\exists y \Diamond \psi]$. Note that we have $\M,w,\sigma'_{[x\mapsto z]} \models \phi[\overline{y}\Diamond \psi]$.

So suppose $\M,w,\sigma_{[x\mapsto z]}\not\models \phi[\exists y \Diamond \psi]$ then we can evaluate the conjunction and disjunctions and obtain some subformula $\chi$ of $\phi[\overline{y}\Diamond \psi]$ such that $\M,w,\sigma'_{[x\mapsto z]} \models \chi$ and\\ $\M,w,\sigma_{[x\mapsto z]}\not\models \hat{\chi}$ where: $\chi = \overline{y}\Diamond \psi$ (in which case $\hat{\chi} = \exists x\Diamond \psi$) or  $\chi  = \hat{\chi}$. Now:

\begin{itemize}
\item If $\chi =  \overline{y}\Diamond \psi$ then 
$\M,w,\sigma'_{[x\mapsto z]} \models ~\bigvee\limits_{i\le l}\Big(~\Diamond~\psi[y_i/y] \lor \Diamond~\psi[y_i'/y]~\Big)$. So there is some $j\le l$ and $\hat{\psi} \in \{ \psi[y_j/y], \psi[y_j'/y]\}$ and $\hat{y_j} \in \{ y_j,y_j'\}$ such that $\M,w,\sigma'_{[x\mapsto z]}\models \Diamond \hat{\psi}$. 

Since $\FV(\hat{\psi}) = \FV(\psi)\cup\{\hat{y_j}\}$, we have $\M,w,\sigma_{[xy\mapsto z\hat{y_j}]} \models \Diamond \psi$. Thus  $\M,w,\sigma_{[x\mapsto z]}\models \exists z \Diamond\psi$ which is a contradiction.

\item Otherwise $\hat{\chi} = \chi$ and also  $y_1,y_1',\ldots y_l,y_l'$ do not occur in $\chi$. Hence $\M,w,\sigma'_{[x\mapsto z]} \models \chi$ iff $\M,w,\sigma_{[x\mapsto z]}\models \chi$  which is a contradiction.

\end{itemize}

($\Leftarrow$)  First note that we can prove a proposition analogous to Proposition \ref{prop-some rule can be applied always} since if $\Gamma$ is not \Esafe then either $(\exists)$ rule can be applied or $(\forall\exists\Diamond)$ rule can be applied. Thus, we can always get a saturated tableau and to show that such a tableau is open, it is sufficient to show that all rules preserve satisfiability.  
We only prove this for $(\forall \exists \Diamond)$ rule (Other rules are already proved in Theorem \ref{thm-tableau for LBF}).

Let $\M,w,\pi\models \forall x \phi[\exists y\Diamond \psi]\land \bigwedge \Gamma$. By lemma \ref{lemma-bounded witness enough}  there exists  $\M'$ and $w' \in \W'$ and $\pi'$ such that $\M',w',\pi' \models \bigwedge \Gamma ~\land~ \exists y_1\exists y_1'\ldots \exists y_l,\exists y_l'~\forall x \phi[\overline{y}\Diamond\psi]$.

\noindent
Let $d_1,d_1',\ldots d_l,d_l' \in \live'(w')$ be the witness for $y_1,y_1',\ldots y_l,y_l'$ respectively and let $\pi''$ be the extension of the assignment $\pi'$ where $y_1y_1',\ldots y_ly_l'$ are assigned to $d_1d_1',\ldots d_ld_l'$ respectively.\\
So we have $\M',w',\pi''\models \bigwedge\Gamma~\land~\forall x \phi[\overline{y}\Diamond\psi]$ as required.

\end{proof}

Note that at if we start with a formula of length $n$ then the application of $(\forall\exists\Diamond)$ rule will blow up the formula to size $2^n$. Thus the size of the tableau is  $2^{O(n^2)}$. So we have a tableau procedure entails an \EXPspace algorithm.

\subsection{Lower bound} \label{sec: lower bound}
The \PSPACE upper bound for the fragments $\AB,\EB$ and  $\BA$ follows  for these fragments since we do not need `clean rewriting' of formulas in the tableau rules specialized for these fragments from propositional modal logic. So the tableau size remains polynomial and hence satisfiability problem for these fragments are \PSPACE-complete.\footnote{$\BE$ fragment needs clean rewriting and hence we get an \EXPspace upper bound.} The $\PSPACE$ lower bound follows from that of propositional modal logic.

We can prove a \NEXPtime lower bound for $\ABEB$ and $\BABE$ fragments over increasing domain models which implies the same lower bound for $\LBF$ and other fragments that contains these as sub fragments.\footnote{The authors of \cite{padmanabha_et_al:LIPIcs:2018:9942} claim \PSPACE upper bound for $\ABEB$ fragment which is rectified here. The bug is that they do not consider clean rewriting in the intermediate steps and hence calculate the tableau size to be polynomial. } 

The \NEXPtime complete version of the tiling problem is that given a tiling instance $\T = (T,H,V,t_0)$ and a natural number $n$ (in unary representation), does there exists a tiling function $f:[0,2^n]^2 \to T$ that can tile the $2^n\times 2^n$ grid \cite{van1996convenience}.

Given a natural number $k\in [0,2^n]$, it has a binary representation with $n$ bits. We use unary predicates $P_0, P_1,\dots, P_{n-1}$ using which $k$ can be encoded. For example, given $n=3$ and $k=5$, the binary representation of $5$ is $101$. It can be encoded by a domain element $d$ at $w$ such that $\M,w \models P_2(d) \land \neg P_1(d) \land P_0(d)$.

First,  we need a set of formulas to force exponential size domain with respect to $n$, and then use this domain to encode the tiling problem. The encoding formulas for $\ABEB$ fragment are given in Figure ~\ref{Fig-ABEB NEXPTIME}. 

\begin{figure}[h]
  \centering
  {\scriptsize
 \begin{tabular}{|r l |}
\hline
$\phi_0:=$& $\exists x~\Box \big(\bigwedge\limits_{0 \le i< n} \neg P_i(x)~\big)$ \\
\hline
$\phi_{1}:=$& $\forallBox{x} \Big( \bigwedge\limits_{0 \le i< n} \big(P_i(x) \implies \Box P_i(x) \land \neg P_i(x) \implies \Box\neg P_i(x)\big)\Big)$\\
\hline

$\phi_{2}:=$& $\forall x\Box \Big(~\bigwedge\limits_{0 \le i < n}\Big(~\neg P_i(x) \implies \exists y_i\Box~ \big( P_i(y) \land \bigwedge\limits_{i\ne j} P_j(x) \leftrightarrow P_j(y)\big)~\Big)~ \Big) $\\
\hline

$\alpha^n :=$&$ \phi_0 ~\land ~\Box^{\le n}~( \phi_1~ \land ~\phi_2 \land \Diamond \top)$\\
\multicolumn{2}{|c|}{}\\
\multicolumn{2}{|l|}{where
 $\Box^{\le n} \psi = \psi \land \Box^{\le n-1} \psi$ ~~~and~~~ $\Box^{\le 0} \psi = \top$}\\
\hline
\hline

$\psi^0:=$&$\forall x \Box\Box\Box \Big(\bigwedge\limits_{0 \le i < n}\big(\neg P_i(x)\big) \implies Q_{t_0}(x,x)~\Big)$\\
\hline
$\psi_1:=$& $\forall x \Box \forall y \Box \Box \Big(  \bigvee\limits_{t} \big(Q_t(x,y) \land \bigwedge\limits_{t'\ne t} \neg Q_{t'}(x,y)~\big)~\Big)$ \\
\hline
$\psi_2:=$& $\forall x \Box \forall y\Box \forall z \Box \Big(succ(x,y)\implies \bigvee\limits_{(t,t')\in H} \big(~ Q_t(x,z) \land Q_{t'}(y,z) ~\big)~ \Big)$\\

$\psi_3:=$& $\forall x \Box \forall y\Box \forall z \Box \Big( succ(x,y) \implies
\bigvee\limits_{(t,t')\in V}\big(~ Q_t(z,x) \land Q_{t'}(z,y) ~\big)~ \Big)$\\
\hline

\multicolumn{2}{|c|}{ }\\

\multicolumn{2}{|l|}{where }\\

$succ(x,y):=$&$\bigvee\limits_{0 \le i < n}\Big( \neg P_i(x) \land P_i(y)~ \land~\bigwedge\limits_{j<i}\big( P_j(x) \land \neg P_j(y)~\big)~$ \\
&\quad\quad\quad\quad\quad\quad\quad\quad\quad$ \land \bigwedge\limits_{j>i}\big( P_j(x)\leftrightarrow P_j(y) ~\big)~\Big)$ \\

\hline
$\beta_T^n :=$&$\alpha^n~ \land~ \Box^n\Big( \phi_1\land \Box \phi_1 \land \Box\Box\phi_1 \land \Diamond^3 \top~ \land  \psi_0 \land \psi_1\land \psi_2 \land \psi_3~\Big)$\\
\hline

\end{tabular}
}
   \caption{$\ABEB$ formulas for encoding \NEXP-time complete tiling problem  over increasing domain models}
  \label{Fig-ABEB NEXPTIME}
\end{figure}

Formula $\phi_0$ says natural number $0$ exists at the root. Formula $\phi_1$ ensures that the encoding propagates to successors. Formula $\phi_2$ asserts that for every position $i$ and every element $x$ if $i^{th}$  bit of $x$ is $0$ then there exists $y$  such that in all the successor worlds, $i^{th}$ bit of $y$ is $1$, and for all $j\neq i$, $j^{th}$ bit of $x$ and $y$ is identical. Formula $\alpha_n$ propagates $\phi_1$ and $\phi_2$ to $n$ depth. 

In a given model $\M$ for all $u,v\in \W$ we say that $v$ is at a distance $n$ from $u$ if there is some path from $u$ to $v$ of length  $n$ (note that $v$ could be at distance $m$ and also at distance $n$ from $u$).

Let $\Sigma^n_j$ be the set of all $n$-length binary strings in which there are at most $j$ many $1$-bit. Clearly, $\Sigma^n_n$ covers all binary strings of length $n$. Also for every string $s$ of length $n$ let $s(j)$ denote the $j^{th}$ position of $s$.

\begin{lemma}
\label{lemma-exponential domain forallbox + existsbox increasing}
The formula $\alpha^n$ is satisfiable. Moreover, for all models $\M$  and $r \in \W$ if $\M,r\models \alpha^n$ then for every $0 \le j < n$ there exists some world $w$ at distance $j$ from $r$ and for every world $w$ at a distance $j$ from $r$, there exists  a one-one function $f_w: \Sigma^n_j \to \live(w)$ such that if $f_w(s) = d$ then for every $0 \le k < n$,\\
 $s(k) = 1$ iff for every $u$ which is a child of $w$, we have $d\in \val(u,P_k)$.
\end{lemma}
\begin{proof}
To prove that $\alpha^n$ is satisfiable, consider the model $\M = (\W,\R,\D,\live,\val)$ where
\begin{itemize}
\item $\W = \{ w_0,w_1\ldots w_n,w_{n+1}\}$
\item $\R = \{ (w_i,w_{i+1}) \mid i\le n\}$
\item $\D = \bigcup\limits_{j=0}^n \Sigma^n_j$ and  $\delta(w_i) = \D$ for all $i\le n+1$
\item For every $P_j$ where $j\le n$ and every $w_i\in W$  define 
 $\rho (w_i,P_j) = \{ s\mid s\in \D$ and $s(j) = 1\}$
\end{itemize}

It can be verified that $\M,w_0 \models \alpha^n$.

For the second part of the lemma, let $\M,r\models \alpha^n$. The proof is by induction on $j$. 

In the base case, $j=0$. Let $d_0 \in \delta(r)$ be the witness for $x$ in $\phi^0$. Define $f_r(0^n) = d_0$. The invariant holds since $\phi_0$ ensures that for all $0 \le k <  n$ and every $r\to u\in \R$  we have $d_0 \not\in \val(u,P_k)$.

\medskip
In the induction step, let $w'$ be any world at distance $j < n$ from $r$. By $\alpha^n$ there is some $w'\to w \in \R^M$ which implies $w$ is at a distance $j+1$ from $r$. 

Now pick any arbitrary $w$ be at a distance $j+1$ from $r$ and let $w'$ be the parent of $w$. By induction hypothesis  there is $f_{w'}: \Sigma^n_j \to \delta(w')$ such that for all $s\in \Sigma^n_j$  if $f_{w'}(s) = d$ then for every position $0 \le k < n$ (since $w$ is the child of $w'$),  $s^k =1$ iff $d\in \rho(w,P_k)$.
Now, define $f_w:\Sigma^n_{j+1} \to \delta(w)$ such that for  every $t\in \Sigma^n_{j+1}$: 

\begin{itemize}
\item If $t\in \Sigma^n_j$ then $f_w(t) = f_{w'}(t)$. 

\item Otherwise, let $k$ be an arbitrary position  such that $k^{th}$ bit of $t$ is  $1$. Consider $t'$ such that for all $l\ne k$ we have $t'(l) = t(l)$ and $t'(k) = 0$. Then, $t'\in \Sigma^n_j$ and let $f_{w'}(t') = d'$. By induction hypothesis (since $w$ is child of $w'$), for every position $j$ we have $t'(j) = 1$ iff $d'\in \rho(w,P_j)$.

Now, by $\alpha^n$ we have $\M,w \models \Diamond \top$. Hence there is at least one child for $w$ and moreover, since $\M,w \models \phi_1$, for every child $u$ of $w$ and for every position $j$ we have $t'(j) = 1$ iff $d\in \val(u,P_j)$. 
\medskip

Now since $\M,w \models \neg P_k(d')$ and $\M,w \models \phi_2$, there is some $d$ (witness for $y_k$) such that for every child $u$ of $w$, we have $\M,u \models P_k(d) \land \bigwedge\limits_{k'\ne k} P_{k'}(d') \leftrightarrow P_{k'}(d)$. \\
Define $f_w(t) = d$.
\end{itemize}

To see that the claim holds, if $t\in \Sigma^n_j$ then the claim holds since $\M,w' \models \phi_1 \land \Box\phi_1$. Otherwise, if $t\in \Sigma^n_{j+1}$ then we have $t'\in \Sigma^n_{j}$ (as defined above) and $t$ and $t'$ differ exactly in the position $k$. Hence by $\M,w\models \phi_2$, the defined $f_w(t)=d$ satisfies the required condition of  the invariant.
\end{proof}

Thus, the exponential tiling problem is encoded in the scope of $n$ modal depth as described in Figure \ref{Fig-ABEB NEXPTIME}. Note that in this encoding we have exponential domain instead of the grid. Thus, we use the binary predicate $Q_t(x,y)$ to say that the grid point $(x,y)$ has tile $t$. The formula $succ(x,y)$ encodes that $y$ is the successor of $x$ ( as a natural number in bit representation).

The formula  $\psi_0$ says $(0,0)$ has tile $t_0$; $\psi_1$ ensures every grid point has exactly one tile; $\psi_2$ and $\psi_3$ encode horizontal and vertical constraints respectively. Also, $\phi_1\land \Box \phi_1 \land \Box\Box\phi_1$ ensures that the `bit encoding' continues for the next $3$ modal depth. The formula $\Diamond^3 \top$ is to ensure that there is some world after $3$ modal depth. Finally, $\beta^n_T$ is the $\ABEB$ formula that encodes the exponential tiling problem. Note that the size of $\beta$ is polynomial in the size of the input instance (since $n$ is given in unary representation).

\begin{theorem} \label{thm-ABEB NEXPTIME}
For all tiling instance $\T= (T,H,V,t_0)$ and all $n\in\mathbb{N}$, there is a proper tiling $f:(2^n\times 2^n)\rightarrow T$ if and only if $\beta^n_T$ is satisfiable. 
\end{theorem}
\begin{proof}
Suppose there is a tiling, then consider the model $\M = (\W,\R,\D,\live,\val)$ where
\begin{itemize}
\item $\W = \{ w_0,\ldots w_n,w_n,w_{n+1} \}~ \cup ~\{ u_0,u_1,u_2\}$
\item $\R = \{ (w_i, w_{i+1}) \mid i \le n\}~ \cup~ \{ (w_{n+1},u_0),~ (u_0,u_1),~ (u_1,u_2) \}$
\item $\D = \{ 0,1,\ldots 2^n\}$ and for all $v\in \W$, $\live(v) = \D$
\item For all $v\in \W$ and for all $0 \le k < n$\\
$\val(v,P_k) = \{ d\mid k^{th}$ bit of $d$ in binary representation is $1\}$ and 
for $u_2$ and every tile $t\in T$ we have\\
$\val(u_2,Q_t) = \{ (c,d) \mid (c,d)$ has the tile $t$ in the tiling$\}$ 
\end{itemize}
It can be verified that $\M,w_0 \models \beta^n_T$.

For the other direction, let $\M,r\models \beta_T^n$. First note that since $\M,r \models \alpha^n$  there exists a $w$ at distance $n$ from $r$ and by Lemma \ref{lemma-exponential domain forallbox + existsbox increasing}, we assume that $\{0,1,\ldots 2^n\} \subseteq \live(w)$ such that and for all $u$ that is a child of $w$ and every $d\in \{ 0,1\ldots 2^n\}$ we have $\M,u \models P_k(d)$ iff $k^{th}$ bit in the binary representation of $d$ in $1$.

Define the tiling $g: \{0,1,\ldots 2^n\}^2 \to T$ where $g(c,d) = t$ iff $\M,w\models \Box\Box\Box Q_t(c,d)$. To prove that $g$ is well defined, note that $\M,w\models \Diamond^3\top$ and also since $\M,w\models \psi_1$ there is a unique tile for every grid point $(c,d)$.

Further, since $\M,w\models \psi_0$, we have $g(0,0) = t_0$. Finally,  $\psi_2$ and $\psi_3$ ensure that the horizontal and vertical tiling constraints are satisfied.
\end{proof}

\begin{corollary}
The satisfiability problem for $\ABEB$ over increasing domain is \NEXPtime-hard. 
\end{corollary}

The encoding formulas of $\BABE$ fragment are very similar. We only need to swap positions of $\Box$ and quantifiers to suit the fragment in the encoding formulas (Fig. \ref{Fig-ABEB NEXPTIME}). The details are omitted. 

\begin{corollary}
The satisfiability problem for $\BABE$ over increasing domain is \NEXPtime-hard. 
\end{corollary}

\section{Conclusion}

We began with the question: are bundles good deals? Now we see the answer is: {\em it depends}, on the combinations of bundles you have chosen, and whether you work with varying domain interpretations or insist on constant domains.

In this paper, we have studied the decidability of bundled fragments of $\FOML$, where we have no restrictions on the use of variables or arity of relations. Over constant domain interpretations, it is only the $\exists \Box$ bundle that is well-behaved, but over increasing domain interpretations, we get a trichotomy of decidability, lack of finite model property and undecidability. The obvious question is to settle the issue of decidability for cases where we have only shown lack of finite model property. This work is one step in the programme of mapping the terrain of decidable bundled fragments and identifying the borderline between decidability and undecidability. We identify some strands that constitute immediate next steps in the programme.

In the context of verification of infinite state systems, security theory and database theory, we are often more interested in the \textsc{model checking} problem. If the domain is finite, this is no different from model checking of first order modal logic. 
However we are usually interested in the specification being checked against a finitely specified (potentially infinite) model, e.g. when the domain elements form a \textit{regular} infinite set. This is a direction to be pursued in the context of bundled fragments.

The applications we have referred to often call for reachability analysis rather than the study of single updates. Moreover, the richness of modal logic stems from its extensions over various classes of frames, and hence the study of bundles over models with various frame conditions is relevant. Unfortunately, while it is clear that equivalence frames seem to lead to undecidability \cite{Wang17d}, even with transitive frames the situation is unclear. Obtaining good decidable fragments over \textsc{linear} frames is an important challenge.

We consider only  the ``pure'' fragments, without constants, function symbols, or equality. The addition of constants is by itself simple, but equality complicates things considerably. Since equality is extensively used in specifications, mapping fragments with equality is an important direction. 

While the decision procedure for the $\exists \Box$ bundle is tight in terms of complexity, the other results involve only an $\EXPspace$ upper bound, leaving a gap between the upper bound and known lower bound. We need sharper technical tools for investigating lower bounds for bundled fragments.

We have presented tableau-based decision procedures that are easily implementable, but inference systems for reasoning in these logics require further study.

\bibliography{ref}

\begin{thebibliography}{10}

\bibitem{AndrekaNB98}
Hajnal Andr{\'{e}}ka, Istv{\'{a}}n N{\'{e}}meti, and Johan van Benthem.
\newblock Modal languages and bounded fragments of predicate logic.
\newblock {\em J. Philos. Log.}, 27(3):217--274, 1998.
\newblock \href {https://doi.org/10.1023/A:1004275029985}
  {\path{doi:10.1023/A:1004275029985}}.

\bibitem{BelardinelliL09}
Francesco Belardinelli and Alessio Lomuscio.
\newblock Quantified epistemic logics for reasoning about knowledge in
  multi-agent systems.
\newblock {\em Artif. Intell.}, 173(9-10):982--1013, 2009.
\newblock \href {https://doi.org/10.1016/j.artint.2009.02.003}
  {\path{doi:10.1016/j.artint.2009.02.003}}.

\bibitem{BelardinelliL12}
Francesco Belardinelli and Alessio Lomuscio.
\newblock Interactions between knowledge and time in a first-order logic for
  multi-agent systems: Completeness results.
\newblock {\em J. Artif. Intell. Res.}, 45:1--45, 2012.
\newblock \href {https://doi.org/10.1613/jair.3547}
  {\path{doi:10.1613/jair.3547}}.

\bibitem{BGG97}
Egon B{\"o}rger, Erich Gr{\"a}del, and Yuri Gurevich.
\newblock {\em The classical decision problem}.
\newblock Springer Science \& Business Media, 2001.

\bibitem{G07FML}
Torben Bra\"uner and Silvio Ghilardi.
\newblock First-order modal logic.
\newblock In P.~Blackburn, J.~van Benthem, and F.~Wolter, editors, {\em
  Handbook of Modal Logic}, pages 549--620. 2007.

\bibitem{DFKLFOtemporal08}
Clare Dixon, Michael Fisher, Boris Konev, and Alexei Lisitsa.
\newblock Practical first-order temporal reasoning.
\newblock In {\em Proceedings of TIME 2008}, pages 156--163, 2008.
\newblock \href {https://doi.org/10.1109/TIME.2008.15}
  {\path{doi:10.1109/TIME.2008.15}}.

\bibitem{GMV99}
Harald Ganzinger, Christoph Meyer, and Margus Veanes.
\newblock The two-variable guarded fragment with transitive relations.
\newblock In {\em Proceedings of LICS ‘99}, pages 24--34, 1999.
\newblock \href {https://doi.org/10.1109/LICS.1999.782582}
  {\path{doi:10.1109/LICS.1999.782582}}.

\bibitem{Hampson2016-HAMDFM}
Christopher Hampson.
\newblock Decidable first-order modal logics with counting quantifiers.
\newblock In {\em Advances in Modal Logic, Volume 11}, pages 382--400. CSLI
  Publications, 2016.
\newblock URL: \url{http://www.aiml.net/volumes/volume11/Hampson.pdf}.

\bibitem{HodkinsonWZ00}
Ian Hodkinson, Frank Wolter, and Michael Zakharyaschev.
\newblock Decidable fragment of first-order temporal logics.
\newblock {\em Ann. Pure Appl. Log.}, 106(1-3):85--134, 2000.
\newblock \href {https://doi.org/10.1016/S0168-0072(00)00018-X}
  {\path{doi:10.1016/S0168-0072(00)00018-X}}.

\bibitem{HodkinsonWZ01}
Ian Hodkinson, Frank Wolter, and Michael Zakharyaschev.
\newblock Monodic fragments of first-order temporal logics: 2000-2001 {A.D}.
\newblock In Robert Nieuwenhuis and Andrei Voronkov, editors, {\em Proceedings
  of {LPAR} '01}, volume 2250, pages 1--23. Springer, 2001.
\newblock \href {https://doi.org/10.1007/3-540-45653-8\_1}
  {\path{doi:10.1007/3-540-45653-8\_1}}.

\bibitem{HWZ02}
Ian Hodkinson, Frank Wolter, and Michael Zakharyaschev.
\newblock Decidable and undecidable fragments of first-order branching temporal
  logics.
\newblock In {\em Proceedings LICS 2002}, pages 393--402, 2002.
\newblock \href {https://doi.org/10.1109/LICS.2002.1029847}
  {\path{doi:10.1109/LICS.2002.1029847}}.

\bibitem{Cresswell96}
G.~E. Hughes and M.~J. Cresswell.
\newblock {\em A New Introduction to Modal Logic}.
\newblock Routledge, 1996.

\bibitem{Liu2019bundled}
Mo~Liu.
\newblock On the decision problems of some bundled fragments of first-order
  modal logic.
\newblock Master's thesis, Peking University, 2019.
\newblock URL: \url{https://arxiv.org/abs/2201.02336}.

\bibitem{padmanabha_et_al:LIPIcs:2018:9942}
Anantha Padmanabha, R~Ramanujam, and Yanjing Wang.
\newblock {Bundled Fragments of First-Order Modal Logic: (Un)Decidability}.
\newblock In {\em Proceedings of FSTTCS 2018}, volume 122, pages 43:1--43:20,
  2018.
\newblock \href {https://doi.org/10.4230/LIPIcs.FSTTCS.2018.43}
  {\path{doi:10.4230/LIPIcs.FSTTCS.2018.43}}.

\bibitem{tendera2019}
Ian Pratt{-}Hartmann, Wieslaw Szwast, and Lidia Tendera.
\newblock The fluted fragment revisited.
\newblock {\em J. Symb. Log.}, 84(3):1020--1048, 2019.
\newblock \href {https://doi.org/10.1017/jsl.2019.33}
  {\path{doi:10.1017/jsl.2019.33}}.

\bibitem{RybakovS17a}
Mikhail~N. Rybakov and Dmitry Shkatov.
\newblock Undecidability of first-order modal and intuitionistic logics with
  two variables and one monadic predicate letter.
\newblock {\em Stud Logica}, 107(4):695--717, 2019.
\newblock \href {https://doi.org/10.1007/s11225-018-9815-7}
  {\path{doi:10.1007/s11225-018-9815-7}}.

\bibitem{van1996convenience}
Peter van Emde~Boas.
\newblock {\em The Convenience of Tilings}.
\newblock Computation and complexity theory: ILLC research report and technical
  notes series. ILLC, 1996.

\bibitem{Vardi1997}
Moshe~Y Vardi.
\newblock Why is modal logic so robustly decidable?
\newblock Technical report, Rice University, 1997.

\bibitem{Wang21}
Xun Wang.
\newblock Completeness theorems for $\exists {\Box}$-fragment of first-order
  modal logic.
\newblock In {\em Proceedings of LORI 2021}, pages 246--258. Springer, 2021.
\newblock \href {https://doi.org/10.1007/978-3-030-88708-7\_20}
  {\path{doi:10.1007/978-3-030-88708-7\_20}}.

\bibitem{Wang17d}
Yanjing Wang.
\newblock A new modal framework for epistemic logic.
\newblock In {\em Proceedings of {TARK} 2017}, pages 515--534, 2017.
\newblock \href {https://doi.org/10.4204/EPTCS.251.38}
  {\path{doi:10.4204/EPTCS.251.38}}.

\bibitem{Wang2018}
Yanjing Wang.
\newblock {Beyond Knowing That: A New Generation of Epistemic Logics}.
\newblock In {\em Outstanding Contributions to Logic}, volume~12, pages
  499--533. Springer Nature, 2018.
\newblock \href {https://doi.org/10.1007/978-3-319-62864-6_21}
  {\path{doi:10.1007/978-3-319-62864-6_21}}.

\bibitem{Wang18}
Yanjing Wang.
\newblock A logic of goal-directed knowing how.
\newblock {\em Synth.}, 195(10):4419--4439, 2018.
\newblock \href {https://doi.org/10.1007/s11229-016-1272-0}
  {\path{doi:10.1007/s11229-016-1272-0}}.

\bibitem{Mono01}
Frank Wolter and Michael Zakharyaschev.
\newblock Decidable fragments of first-order modal logics.
\newblock {\em J. Symb. Log.}, 66(3):1415--1438, 2001.
\newblock URL: \url{http://www.jstor.org/stable/2695115}.

\end{thebibliography}

\end{document}